\documentclass[11pt,a4paper]{article}

\makeatletter
\def\input@path{{./}{./macros/}{./paper/tex/}{./paper/tex/macros/}}
\makeatother

\usepackage[utf8]{inputenc}
\usepackage[T1]{fontenc}
\usepackage[english]{babel}
\usepackage{lmodern}
\usepackage{microtype}
\microtypesetup{protrusion=false}
\microtypecontext{protrusion=basictext}

\usepackage{amsmath,amssymb,amsthm}
\usepackage{mathtools}
\usepackage{fvextra}
\usepackage{csquotes}
\usepackage[a4paper,margin=28mm]{geometry}

\usepackage[hidelinks,bookmarks=false,unicode,hypertexnames=false]{hyperref}

\usepackage[nameinlink,capitalise]{cleveref}

\crefname{theorem}{Theorem}{Theorems}
\crefname{lemma}{Lemma}{Lemmas}
\crefname{proposition}{Proposition}{Propositions}
\crefname{corollary}{Corollary}{Corollaries}
\crefname{definition}{Definition}{Definitions}
\crefname{remark}{Remark}{Remarks}
\crefname{section}{Sec.}{Secs.}

\usepackage{tikz}
\usetikzlibrary{arrows.meta,positioning,fit,calc}

\hypersetup{
  colorlinks=false,
  pdfborder={0 0 0},
  pdfencoding=auto,
  psdextra
}
\usepackage{caption}
\usepackage{thmtools}
\usepackage{booktabs}
\usepackage{graphicx}
\usepackage{tabularx}
\usepackage{makecell}
\usepackage{doi}
\usepackage{enumitem}

\makeatletter
\def\input@path{{macros/}}
\makeatother
\usepackage{iecz}

\providecommand{\VC}[2]{V_{\CalC}(#1,#2)}

\usepackage{datetime2}
\usepackage{xcolor}
\newcolumntype{Y}{>{\raggedright\arraybackslash}X}
\allowdisplaybreaks

\newcommand{\printorcid}[1]{%
  \textcolor{black}{\href{https://orcid.org/#1}{\texttt{#1}}}%
}

\newcommand{\PaperTitle}{IECZ-II: Minimax, Gadget Composition, and SoS Dichotomy}
\newcommand{\PaperSubtitle}{Window-stable Reductions 3XOR\,$\to$\,3SAT and Size-aware MDL}

\title{\PaperTitle\thanks{%
Artifacts: Zenodo v1.0.0 (\href{https://doi.org/10.5281/zenodo.17220388}{10.5281/zenodo.17220388}); 
IECZ-I bundle v1.0.1 (\href{https://doi.org/10.5281/zenodo.17141362}{10.5281/zenodo.17141362}).}%
\\[0.25ex]\large \PaperSubtitle}  
\author{Marko Lela\\[0.3ex]\printorcid{0009-0008-0768-5184}}
\date{\DTMdate{2025-09-29}}

\hypersetup{
  pdftitle   = {\PaperTitle},
  pdfauthor  = {Marko Lela (ORCID: 0009-0008-0768-5184)},
  pdfsubject = {Size-aware MDL, window-stable reductions, SoS},
  pdfkeywords= {3XOR, 3SAT, MDL, SoS, reductions, lower bounds}
}

\declaretheorem[name=Theorem,numberwithin=section]{theorem}
\declaretheorem[name=Lemma,sibling=theorem]{lemma}
\declaretheorem[name=Corollary,sibling=theorem]{corollary}
\declaretheorem[name=Proposition,sibling=theorem]{proposition}
\declaretheorem[name=Definition,style=definition,numberwithin=section]{definition}
\declaretheorem[name=Remark,style=remark,numberwithin=section]{remark}

\newenvironment{gadgetiface}[1]{\par\medskip\noindent\textbf{#1}\par\begin{description}}{\end{description}\medskip}
\newcommand{\ifaceline}[2]{\item[\textbf{#1}] #2}



\usepackage{fancyhdr}
\ifcsname IECZII\endcsname
\else
\fi
\ifcsname RepoTag\endcsname
  \hypersetup{
    pdfsubject = {IECZ umbrella project; repo \RepoTag},
    pdfkeywords= {3XOR, 3SAT, MDL, SoS, reductions, lower bounds, IECZ, \RepoTag}
  }
\fi

\begin{document}
\maketitle

\begin{abstract}
We give a size-aware minimum description length (MDL) cost and a window-stable \XORtoSAT{} chain with controlled losses in total variation (TV) and instance size.
Core components are minimax monotonicity (\Cref{thm:minimax}), gadget composition (XOR and Index; \Cref{thm:gadget}), and a sum-of-squares (SoS) entropy dichotomy (\Cref{thm:sos-dichotomy}) at $k=\Theta(\log n)$.
We present a composition pipeline (E1–E4) for distributional transfers with explicit accounting for TV, size, and meta overheads.
For symbol conventions and standing assumptions, see Sec.~\ref{sec:assumptions}.
\emph{All guarantees are distributional (on a calibrated window); we make no worst-case claims.}

\smallskip
\noindent\emph{Relation to IECZ-I.}
This sequel builds on our companion work \cite{Lela2025IECZ}, which introduced the size-aware MDL functional and the distribution-preserving reduction principle (A2d) together with a $\delta\!=\!0$ injective 3XOR$\to$3SAT translation; here we develop the minimax principle, gadget composition, and the SoS dichotomy.
\end{abstract}

\section{Introduction}
We motivate the IECZ line: stability under deterministic reductions that are $1$-Lipschitz in total variation, prefix-free preprocessing, and a window discipline that permits measure-faithful transfers.
As formalized in the first part (IECZ-I), we adopt the same window convention and prefix-free meta-encoding throughout \cite{IECZI}.
We use a windowed composition pipeline (E1–E4) for \emph{distributional} transfers, formulated as standalone lemmas with explicit TV/size/meta accounting (see \Cref{lem:arrow1,lem:arrow2,lem:arrow3,lem:arrow4} and \Cref{fig:pipeline-e1-e4}).
For empirical sanity checks of the injective \XORtoSAT{} mapping, see Appendix~\ref{app:artifact}.
All statements are windowed: deterministic maps are $1$-Lipschitz in total variation on calibrated windows, and prefix-free meta-encodings add only $O(\log N)$ overhead.
We do not claim worst-case statements here.

\paragraph{Scope and limitations.}
All statements in this paper are \emph{window-calibrated} and hence distributional: our total-variation (TV) guarantees hold on the image window~$\win$ and may degrade off-window; we do not claim worst-case hardness.
Results are stated in a parameterized regime (e.g., $d=O(\log n)$ and $k=\Theta(\log n)$). Varying $c\in(0,1)$ trades bottom width $\W$ against depth $d$ along the switching path; the constants propagate through E1–E4 without changing asymptotics.
For the low-degree and SoS analyses we use absolute constants $C,\tau$; these, as well as the degree budgets, are tied to the window calibration and to the gadget-induced product proxies, and we state them explicitly where they are used.
Superlogarithmic regimes lie outside our present scope and are left to follow-up work.

\paragraph{On the $\boldsymbol{d,k=\Theta(\log n)}$ regime.}
Our log-degree regime is deliberate: it is the natural sweet spot where (a) the switching-lemma width control remains uniform along our calibrated path, (b) the windowed hypercontractive constants can be taken absolute, and (c) the \SoS{} entropy dichotomy yields quantitative degree lower bounds with explicit constants. Outside this regime (e.g., superlogarithmic depths), the bottom width $\W$ may grow with $d$ via the dependence on $p=m^{-c/d}$, and our absolute constants $(C,\tau)$ are no longer claimed. We therefore do \emph{not} assert worst-case or non-parameterized statements in this paper.

\paragraph{Roadmap and companion sequel.}
A companion sequel (IECZ-III, in preparation) develops a \emph{hardcore–condensation–lift} pipeline that
\begin{enumerate}[label=(\roman*), leftmargin=*, itemsep=.25\baselineskip]
  \item performs blockwise aggregation within~$\win$,
  \item selects a TV-stable hardcore sub-window,
  \item applies a seeded condensation/re-encoding with explicit TV/Size/Meta bookkeeping, and
  \item derandomizes to an explicitly parameterized distributional family~$\{\Phi_n\}$.
\end{enumerate}
The aim is to carry the IECZ-II minimax–SoS dichotomy through to size-aware hardness for explicit families,
while remaining in a distributional, non-worst-case setting; see~\cite{IECZIII}.

\paragraph{Parameter preview (E1–E4).}
We commit to concrete choices used later in \Cref{thm:blueprint}:
\[
  \begin{aligned}
    t      &= \Bigl\lceil c_1\,\varepsilon^{-2}\log N \Bigr\rceil,\\
    \theta &\ge c_2 \quad \text{(hardcore mass)},\\
    s      &= O(\log N) \quad \text{(condenser seed length)}.
  \end{aligned}
\]
These yield product amplification to advantage $N^{-c_1'}$ (E1), a constant-mass
hardcore with advantage $N^{-c}$ (E2) via~\cite{Impagliazzo1995-Hardcore},
preservation under seeded condensation with polylogarithmic size overhead (E3)
using extractor-style condensers~\cite{Trevisan2001-Extractors}, and an
\emph{explicitly parameterized distributional} family after fixing indices with
$O(\log N)$ meta-overhead (E4).

We first formalize a size-aware MDL cost and prove \emph{minimax monotonicity}
(\Cref{thm:minimax}). We then establish a window-stable injective reduction
\XORtoSAT{} with tabulated losses, enabling a clean transfer corollary.
Next, we analyze \emph{gadget composition} with width control along the switching path $\ppath$ (\Cref{thm:gadget}).
Finally, we state an \SoS{} \emph{entropy dichotomy} that either yields a compression echo or forces degree lower bounds
(\Cref{thm:sos-dichotomy}).

\paragraph{Related work (concise).}
This paper is a sequel to IECZ-I, which introduced the window discipline and size-aware MDL viewpoint we use here \cite{IECZI}.
Our calibrated-window, size-aware pipeline builds on classic tools while tracking TV and size overheads.
For gadget composition and switching arguments we follow Håstad’s switching-lemma methodology with
explicit width control in distributional settings~\cite{Hastad1986-Switching, Hastad01}. Our SoS statements align
with the modern low-degree and SoS viewpoint for CSPs~\cite{BarakEtAl2012-SOS} and the Fourier/entropy
toolkit as surveyed in~\cite{ODonnell2014-AoBF}. The injective \XORtoSAT mapping parallels Tseitin-style encodings
but is window-stable and auxiliary-free, with TV losses tracked (cf.\ classical clause gadgets~\cite{Tseitin68,ImpagliazzoPZ01}).
Our MDL/minimax lens connects to coding-aware reductions and prefix-free overhead analyses in learning/complexity~\cite{Rissanen78, Grunwald2007-MDL}.


\section{Central assumptions and notation}\label{sec:assumptions}
We collect the global conventions referenced across sections.
\emph{We follow the conventions introduced in IECZ-I} \cite{Lela2025IECZ} \emph{and restate only what is needed here.}

\paragraph{Distributions and total variation.}
All distributions live on finite instance spaces with canonical encodings.
For distributions $\mu,u$ over the same space, $\TV(\mu,u)$ denotes total-variation distance.
Deterministic maps $R$ act by pushforward $R\push\mu$, and total variation contracts under deterministic postprocessing:
$\TV(R\push\mu,\,R\push u) \le \TV(\mu,u)$.
This aligns with the distribution-preserving reduction principle (A2d) from \cite{Lela2025IECZ}.

\paragraph{Prefix-free coding convention.}
Preprocessors/reductions (and any meta-data such as seeds/indices) are encoded prefix-free.
Their description length contributes an additive $O(\log N)$ overhead to costs; $N$ is the instance bitlength.

\paragraph{Size and windows.}
$\size(\cdot)$ is the bitlength of the instance encoding.
A \emph{window} $\win$ is an image slice where the reduction is measure-faithful:
we write $\dtv := \TV(R\push\mu,\,u)$, and on $\win$ one has $\dtv=0$, with explicitly bounded size growth
(as in \cite{Lela2025IECZ}).

\paragraph{Cost functional.}
$V_{\CalC}(\mu,\varepsilon)$ denotes a size-aware MDL-type cost for class $\CalC$
at advantage/error parameter $\varepsilon$; all statements are with respect to the fixed coding convention above.
We adopt the definition and accounting rules introduced in \cite{Lela2025IECZ}.

\paragraph{Switching path and width.}
We fix the depth $d=O(\log n)$ and parameter path $\ppath$, implying bottom width
$\W = m^{O(c/d)}$ in the adopted switching-lemma regime.

\paragraph{Balancing.}
Where required, the input distribution is balanced on $\win$ (specified per section).

\paragraph{Standing assumptions (global).}
Prefix-free meta-encodings are fixed (additive $O(\log N)$ overhead).
Windows $\win$ are measure-faithful ($\dtv=0$) with controlled size growth.
Balancing and moment-matching conditions are enforced only where explicitly invoked.
The switching path $\ppath$ is calibrated for target depth $d=O(\log n)$ with bottom width $\W$.

\subsection*{Notation table (symbols used throughout)}
\begin{table}[t]
  \centering
  \caption{Symbols and their meaning (windowed setting).}
  \small 
  \renewcommand{\arraystretch}{1.08}
  \begin{tabularx}{\linewidth}{@{}p{0.24\linewidth} Y@{}}
    \toprule
    \textbf{Symbol} & \textbf{Meaning} \\ \midrule
    $\mu,u$ & Target/windowed distribution and proxy/structured distribution on the same instance space. \\
    $\TV(\mu,u)$ & Total-variation distance between $\mu$ and $u$ (computed on $\win$ when specified). \\
    $R\push\mu$ & Pushforward of $\mu$ by a deterministic map $R$ (reduction/preprocessor). \\
    $\win$ & Image window where the reduction is measure-faithful ($\dtv=0$) and size growth is bounded. \\
    $\size(\cdot)$ & Bitlength of an instance encoding; used for size-blowup accounting. \\
    $\CalC$ & Computational/model class (closed under poly-time preprocessing). \\
    $\VC{\mu,\varepsilon}$ & Size-aware MDL-type cost of class $\CalC$ at parameter $\varepsilon$. \\
    $d$ & Switching depth, chosen $d=O(\log n)$. \\
    $p$ & Path parameter in the switching regime, $p=m^{-c/d}$ (fixed $c\in(0,1)$). \\
    $W$ & Bottom width bound, $W=\mathrm{poly}(1/p)=m^{O(c/d)}$. \\
    $\Poly[k]$ & Real polynomials of degree $\le k$ on instance coordinates. \\
    $\Delta_k(\mu,u)$ & Low-degree test discrepancy: $\sup\{|\E_\mu[p]-\E_u[p]|: p\in\Poly[k], \|p\|_\infty\le 1\}$. \\
    $\tE$ & Degree-$k$ \SoS{} pseudoexpectation calibrated to $u$. \\
    $\Delta^{\mathrm{sos}}_k(\mu,u)$ & \SoS{}-based discrepancy: $\sup\{|\E_\mu[p]-\tE[p]|: p\in\Poly[k], \|p\|_\infty\le 1\}$. \\
    $N$ & Instance bitlength used in prefix-free meta-encoding overheads ($O(\log N)$). \\
    \bottomrule
  \end{tabularx}
\end{table}

\section{Notation and standing conventions}\label{sec:notation}
We work over the Boolean cube $\{0,1\}^n$ unless noted otherwise.
Expectations, variances, and covariances under a distribution $\mu$ are
$\E_\mu[\cdot]$, $\Var_\mu(\cdot)$, and $\Cov_\mu(\cdot,\cdot)$.
Total variation is $\TV(\mu,u)$.
Deterministic maps $R$ act on measures by pushforward $R\push\mu$
and satisfy contraction: $\TV(R\push\mu, R\push u)\le \TV(\mu,u)$.
The image \emph{window} is denoted by $\win$ (as introduced in IECZ-I);
unless stated otherwise, all expectations and pseudoexpectations are taken on $\win$.

\paragraph{Size and encoding.}
$\size(\cdot)$ denotes the bitlength of the (canonical, prefix-free) instance encoding.
Meta/preprocessing descriptions are prefix-free and contribute an additive $O(\log N)$
overhead, where $N$ is the instance bitlength.

\paragraph{Low degree and SoS.}
Low-degree polynomials are $p\in\Poly[k]$ (total degree $\le k$).
For sum-of-squares we write $\tE[\cdot]$ for a degree-$k$ pseudoexpectation
calibrated to the proxy $u$.

\paragraph{Biased Walsh basis (relative to a product proxy $u$).}
If $u$ is a product distribution on $\{0,1\}^n$ with biases $b_i=\E_u[x_i]\in(0,1)$, set
\[
\chi_i(x_i) \coloneqq \frac{x_i-b_i}{\sqrt{b_i(1-b_i)}},\qquad
\chis{S}(x)\coloneqq \prod_{i\in S}\chi_i(x_i)\quad(S\subseteq[n]).
\]
Then $\{\chis{S}\}_{S\subseteq[n]}$ is orthonormal in $L^2(u)$, so
$p=\sum_{|S|\le k}\hat p(S)\chis{S}$ and
$\|p\|_{2,u}^2=\sum_{|S|\le k}\hat p(S)^2=\E_u[p^2]$.
We use $\langle f,g\rangle_u\coloneqq \E_u[fg]$ for the $u$-inner product.

\paragraph{Discrepancy measures (used later).}
For $k\ge 0$, the low-degree test discrepancy and its SoS analogue are
\[
\begin{aligned}
\Delta_k(\mu,u)
&\coloneqq \sup\bigl\{\,|\E_\mu[p]-\E_u[p]|:
  p\in\Poly[k],\ \|p\|_\infty\le 1\,\bigr\},\\
\Delta^{\mathrm{sos}}_k(\mu,u)
&\coloneqq \sup\bigl\{\,|\E_\mu[p]-\tE[p]|:
  p\in\Poly[k],\ \|p\|_\infty\le 1\,\bigr\}.
\end{aligned}
\]
When we write $\dtv$, we mean $\dtv\coloneqq \TV(R\push\mu,u)$
(on $\win$ if a window is fixed).

\paragraph{Problems and reductions.}
We write \XOR{}, \SAT{}, \CLIQUE{} for problems, and name reductions as
$R_{\text{xor}\to\text{sat}}$, $R_{\text{sat}\to\text{clique}}$.
Composition is $R_2\circ R_1$; size blowups are stated explicitly (e.g.\ $m'=4m$).
Asymptotics are with respect to $n$ unless another parameter is specified.
Degree lower bounds for \SoS{} refer to the standard pseudoexpectation model.

\paragraph{Standing assumptions (global, windowed).}
Windows $\win$ are measure-faithful for the stated maps (i.e.\ $\dtv=0$ on $\win$)
with controlled size growth; balancing/moment matching is enforced only where explicitly
invoked; switching-path parameters and $k=\Theta(\log n)$ are as specified in the
corresponding sections.

\section{Framework: Size-aware MDL and Windows}\label{sec:framework}
\noindent\textit{Consistency with IECZ-I.} We adopt the same prefix-free coding model and the same window discipline as IECZ-I~\cite{IECZI}; definitions are restated here for completeness.
 
\subsection{Coding model and decodability}
\label{sec:coding-model}
We work in a fixed, machine-independent prefix-free coding model (as in IECZ-I~\cite{IECZI}), using self-delimiting integer codes in the spirit of Elias~\cite{Elias1975-Universal}.
A \emph{model} $M\in\CalC$ together with a deterministic preprocessor $R$
is represented by a prefix-free description $\code(M,R)\in\{0,1\}^\ast$ such that:
\begin{enumerate}[label=(\roman*),ref=(\roman*)]
  \item \textbf{Decodability.} There is a uniform decoder $\dec$ with
        $\dec(\code(M,R))=(M,R)$; concatenations use self-delimiting integers
        encoded with Elias--$\delta$; all lengths are counted in bits~\cite{Elias1975-Universal}.
  \item \textbf{Kraft admissibility.} The description family is prefix-free:
        $\sum_{w\in\pf}2^{-\len(w)}\le 1$.
  \item \textbf{Instance length.} We write $N=\size(x)$ for the bitlength of an instance $x$.
\end{enumerate}
The \emph{meta-overhead} for specifying $R$ and integer parameters (e.g.\ seeds, indices,
window identifiers) is additive and, under Elias--$\delta$, satisfies
$\len(\code(R)) \le \bigl(\sum_{i=1}^{m} k_i + 1\bigr)\log_2 N$ for all $N\ge 16$
(cf.\ \cref{lem:prefix-overhead}); see also the MDL perspective~\cite{Grunwald2007-MDL}.

\begin{lemma}[Prefix-free meta-overhead (explicit $\delta$-code)]\label{lem:prefix-overhead}
Fix Elias--$\delta$ for all integer fields.
Let $R$ be any polynomial-time map on instances of length $N$, whose meta-parameters $t_1,\dots,t_m$ each lie in a range $[0,T_i(N)]$ with $T_i(N)\le N^{k_i}$ for some constants $k_i\ge 0$.
Then, writing $\ell_\delta(\,\cdot\,)$ for the Elias--$\delta$ length in bits,
\[
  \ell_\delta(t_i)\ \le\ \lfloor \log_2 \max\{1,t_i\}\rfloor\ +\ 2\,\big\lfloor \log_2\bigl(1+\lfloor\log_2 \max\{1,t_i\}\rfloor\bigr)\big\rfloor\ +\ 1,
\]
and hence, for all $N\ge 16$ and all $i$ with $t_i\le T_i(N)\le N^{k_i}$,
\[
  \ell_\delta(t_i)\ \le\ k_i\,\log_2 N\ +\ 2\log_2\log_2 N\ +\ 3\ \ \le\ \ (k_i{+}1)\,\log_2 N.
\]
Consequently, for any constant-arity template (constant $m$) one has the total meta-overhead
\[
  \len(\code(R))\ =\ \sum_{i=1}^{m}\ell_\delta(t_i)\ +\ O(1)\ \le\ \Bigl(\sum_{i=1}^{m}k_i + 1\Bigr)\,\log_2 N\ =\ O(\log N).
\]
The same bound holds for any finite tuple of such parameters encoded by prefix-free concatenation.
\end{lemma}

\begin{proof}
For Elias--$\delta$, the length of a positive integer $t\ge 1$ is
\(
  \ell_\delta(t)= \lfloor \log_2 t \rfloor \ +\ 2\lfloor \log_2(1+\lfloor \log_2 t \rfloor) \rfloor \ +\ 1,
\)
see~\cite{Elias1975-Universal}. Allow $t=0$ by coding $\max\{1,t\}$ (adds at most one bit).
If $t\le N^{k_i}$, then $\lfloor \log_2 t \rfloor \le k_i\log_2 N$, and
$\log_2(1+\lfloor\log_2 t\rfloor)\le \log_2\log_2 N^{k_i} \le \log_2\log_2 N + \log_2 k_i \le \log_2\log_2 N + 1$
(for fixed $k_i$ and all $N\ge 16$). This yields
\(
  \ell_\delta(t_i)\ \le\ k_i\log_2 N + 2\log_2\log_2 N + 3\ \le\ (k_i{+}1)\log_2 N.
\)
Summation over a constant number $m$ of fields and adding the constant template/tag overhead gives the displayed bounds.
Prefix-freeness is by construction of Elias--$\delta$ and stability under concatenation of self-delimiting fields.
\end{proof}

\subsection{Window discipline}

\begin{lemma}[TV is 1-Lipschitz under deterministic maps]\label{lem:tv-contraction}
For any distributions $\mu,u$ on the same finite space and any deterministic map $R$,
\[
  \TV(R\push\mu, R\push u) \le \TV(\mu,u).
\]
\end{lemma}

\begin{proof}
Write $\TV(\mu,u)=\sup_{A}(\mu(A)-u(A))$, where the supremum is over measurable sets.
For any measurable $B$ in the codomain, $(R\push\mu)(B)- (R\push u)(B) = \mu(R^{-1}(B))-u(R^{-1}(B))$,
so taking the supremum over $B$ is a supremum over preimages $R^{-1}(B)$. Hence the supremum cannot increase:
\(\TV(R\push\mu, R\push u) \le \TV(\mu,u)\).
\end{proof}

A \emph{window} $\win$ is a measurable image slice on which a fixed reduction is
\emph{approximately} measure-faithful in total variation. We write
\[
  \dtvwin \;:=\; \TV\!\bigl(R\push\mu,\ u\bigr) \;=\; o(1)\quad\text{on }\win,
\]
uniformly along the calibrated blueprint path. (In the special case of exact faithfulness one has
$\dtvwin=0$; all our bounds are stated for $\dtvwin=o(1)$.) Size growth is explicitly bounded
(e.g.\ clause blowup $m'=4m$ for \XOR{}$\to$\SAT{}, cf.\ \cite{IECZI}). All later statements are parameterized by window
balance and boundary error; see the loss tables in each section.

\begin{remark}[Distributional framing]
All complexity and \SoS{} statements in this paper are interpreted on the calibrated image window~$\win$.
The sequel~\cite{IECZIII} takes this framing as input to a lift step (aggregation $\to$ hardcore selection
$\to$ seeded condensation $\to$ derandomization) that targets parameterized distributional families
with controlled TV/size/meta overheads.
\end{remark}

\begin{remark}[Width calibration across regimes]
Throughout the calibrated path we take $p=m^{-c/d}$ with $d=\Theta(\log n)$ and a fixed small constant $c\in(0,1)$, which yields bottom width $\W=\poly(1/p)$ uniformly on the window. This uniformity is what underpins the absolute hypercontractive constants used later. If $d$ is allowed to grow superlogarithmically, the same relation forces smaller $p$ and can drive $\W$ beyond polynomial ranges; our analysis does not track this deterioration. All quantitative claims here are therefore scoped to $d,k=\Theta(\log n)$; see also Sec.~\ref{sec:gadgets}.
\end{remark}

\paragraph{Style \& symbols (conventions).}
\begin{itemize}
  \item \textbf{Window.} We write $\win$ in math and “Win” in prose for the image slice on which the reduction is measure-faithful.
  \item \textbf{Pushforward.} Deterministic maps act by pushforward $R\push\mu$; total variation (TV) is $1$-Lipschitz: $\TV(R\push\mu,\ R\push u)\le \TV(\mu,u)$.
  \item \textbf{Fields/groups.} We write $\Ftwo$ for the binary field; $\GL_2(\Ftwo)$ denotes the group of invertible $2{\times}2$ matrices over $\Ftwo$.
  \item \textbf{Characters.} $\chis{S}$ denotes the Walsh character on support $S$ (orthonormal under the proxy on $\win$).
  \item \textbf{Noise.} $T_\theta$ (Bonami–Beckner) acts diagonally on degree-$d$ slices with eigenvalues $\theta^d$.
\end{itemize}

\paragraph{Notation at a glance.}
Unless stated otherwise, all probabilistic statements are understood on~$\win$.
We use the shorthands
\[
  \TV(\mu,u)\ \text{(total variation)},\qquad
  R\push\mu\ \text{(pushforward)},\qquad
  \dtvwin := \TV(R\push\mu,u),
\]
\[
\begin{aligned}
  d,k &\ \text{(degree budgets)}\\
  T_\theta &\ \text{(Bonami--Beckner)}\\
  \chis{S} &\ \text{(Walsh character)}\\
  \baseline_\win &\ \text{(baseline conditioned to }\win\text{)}
\end{aligned}
\]

\subsection{Definition of the size-aware MDL cost}
\label{sec:mdl-cost}
Fix a class $\CalC$ (closed under polynomial-time preprocessing) and an error/advantage
parameter $\varepsilon\in[0,1]$. For a distribution $\mu$ on instances of length $N$ define
\begin{align}\label{eq:VC-def}
  V_{\CalC}(\mu,\varepsilon)
  \ :=\ 
  \inf_{M\in\CalC,\ R\ \mathrm{poly}}\ 
  \Bigl\{
      \len(\code(M,R))
      \;+\;
      \E_{x\sim\mu}\!\bigl[L_M(R(x),\varepsilon)\bigr]
  \Bigr\}.
\end{align}
where $L_M(\cdot,\varepsilon)$ is a nonnegative per-instance loss measuring the
residual description under $M$ at target error $\varepsilon$
(e.g.\ a codelength for mistakes or a hinge-type penalty; concrete form not
needed for the monotonicity arguments).
The coding convention is fixed once and for all; hence all $O(\log N)$ terms
are absolute with respect to \cref{lem:prefix-overhead}.
This matches the IECZ-I definition~\cite{IECZI}; we restate it for completeness.

\begin{remark}[Role of $L_M$]
Our results require only that $L_M$ be subadditive under concatenation,
normalized to vanish on perfect fits, and insensitive to $O(1)$ renamings.
These mild conditions hold for standard MDL-style penalties and for prediction
losses that translate to codelength by Shannon coding.
\end{remark}

\begin{lemma}[Subadditivity of the per-instance loss]\label{lem:L-subadd}
Let $L_M(\cdot,\varepsilon)\ge 0$ be the per-instance loss for a fixed model $M\in\CalC$ at target $\varepsilon$.
Assume the coding convention of Sec.~\ref{sec:coding-model}. Then for any two instances $x,y$ (over disjoint supports or
properly tagged concatenations) one has
\[
  L_M(x\Vert y,\varepsilon) \;\le\; L_M(x,\varepsilon) \;+\; L_M(y,\varepsilon) \;+\; O(1),
\]
where $\Vert$ denotes the canonical concatenation in our encoding. The $O(1)$ term accounts for the
self-delimiting tag separating the two parts and is independent of $N$.
\end{lemma}

\begin{proof}
Concatenate the two descriptions produced by $M$ on $x$ and $y$ and add a self-delimiting boundary tag.
By prefix-free decodability (Sec.~\ref{sec:coding-model}) the decoder can split and feed each part to the same $M$.
The score/codelength is additive up to the single boundary tag, which contributes $O(1)$ bits
(self-delimiting constant field). Nonnegativity is by definition.
\end{proof}

\begin{definition}[Model class assumptions]\label{def:model-class}
Throughout, the model class $\CalC$ satisfies:
\begin{enumerate}
  \item \textbf{Closure under poly-time preprocessing.} If $M\in\CalC$ and $R$ is deterministic poly-time, then the composed system $(x\mapsto M(R(x)))$ is represented in $\CalC$ with prefix-free description length $\ell(\code(M,R))=\ell(\code(M))+\ell(\code(R))+O(1)$.
  \item \textbf{Uniform decoding.} There is a single uniform decoder $\mathrm{decode}$ that maps any description word to $(M,R)$; see Sec.~\ref{sec:coding-model} for the coding model.
  \item \textbf{Stable per-instance loss.} The loss $L_M(\cdot,\varepsilon)\ge 0$ is invariant under $O(1)$ renamings and obeys Lemma~\ref{lem:L-subadd}.
\end{enumerate}
These assumptions are exactly what Theorems~\ref{thm:minimax} and \ref{thm:blueprint} use.
\end{definition}

\subsection{Minimax monotonicity (formal)}
We restate \Cref{thm:minimax} with the explicit coding split.

\begin{theorem}[Minimax monotonicity]\label{thm:minimax}
Let $R$ be deterministic polynomial time and $\CalC$ be closed under polynomial-time preprocessing.
Then for all $\mu$ and $\varepsilon$,
\[
  V_{\CalC}(R\push \mu,\varepsilon)
  \ \le\
  V_{\CalC}(\mu,\varepsilon) \;+\; \len(\code(R)).
\]
In particular, under our fixed coding model, $\len(\code(R)) = O(\log N)$ by \Cref{lem:prefix-overhead}.
\end{theorem}

\begin{proof}
Let $V_{\CalC}(\mu,\varepsilon)$ be defined as in \eqref{eq:VC-def} (cf.\ \eqref{eq:VC-def} in Sec.~\ref{sec:mdl-cost}):
\[
  V_{\CalC}(\mu,\varepsilon)
  \;=\;
  \inf_{M\in\CalC,\;R_\circ\ \mathrm{poly}}\ 
  \Bigl\{\, \len(\code(M,R_\circ))\ +\ \E_{x\sim\mu}\bigl[L_M(R_\circ(x),\varepsilon)\bigr] \Bigr\}.
\]
Fix $\delta>0$ and choose $(M^\star,R_\circ^\star)$ $\delta$-optimal for $\mu$, so that
\[
  V_{\CalC}(\mu,\varepsilon)\ \ge\ 
  \len(\code(M^\star,R_\circ^\star)) + \E_\mu\!\left[L_{M^\star}(R_\circ^\star(x),\varepsilon)\right] - \delta.
\]

Consider reusing the same preprocessor on the image space: define $R_\circ^{\star\prime}:=R_\circ^\star$ (as an arbitrary poly-time map on bitstrings, it can be applied to $y$ directly). By change of variables under the pushforward,
\[
  \E_{x\sim\mu}\!\left[L_{M^\star}\!\bigl(R_\circ^\star(R(x)),\varepsilon\bigr)\right]
  \;=\;
  \E_{y\sim R\push\mu}\!\left[L_{M^\star}\!\bigl(R_\circ^{\star\prime}(y),\varepsilon\bigr)\right].
\]
with $R_\circ^{\star\prime}(R(x))=R(R_\circ^\star(x))$ on the image.

Now encode $(M^\star, R_\circ^{\star\prime})$ together with $R$ in the prefix-free model. By \Cref{lem:prefix-overhead},
$\len(\code(R))=O(\log N)$ and concatenation is prefix-free, hence
\begin{multline*}
  V_{\CalC}(R\push\mu,\varepsilon)
  \;\le\;
  \len\!\bigl(\code(M^\star, R_\circ^{\star\prime})\bigr)
  \;+\;
  \E_{\,y\sim R\push\mu}\!\Bigl[L_{M^\star}\!\bigl(R_\circ^{\star\prime}(y),\varepsilon\bigr)\Bigr]
  \\[2pt]
  \;\le\;
  \len\!\bigl(\code(M^\star, R_\circ^\star)\bigr)
  \;+\;
  \E_{\,x\sim \mu}\!\Bigl[L_{M^\star}\!\bigl(R_\circ^\star(x),\varepsilon\bigr)\Bigr]
  \;+\; \len(\code(R))\,.
\end{multline*}
Taking infimum over $(M^\star,R_\circ^\star)$ and letting $\delta\to 0$ yields
$V_{\CalC}(R\push\mu,\varepsilon)\ \le\ V_{\CalC}(\mu,\varepsilon)+\len(\code(R))$, as claimed. By \Cref{lem:prefix-overhead}, $\len(\code(R))=O(\log N)$ under our fixed coding model.
\end{proof}

\subsection*{Tightness up to \(O(\log N)\)}
\begin{proposition}[Tightness]\label{prop:tight}
There exist families $(\mu_N)$ and maps $R_N$ for which the gap
\(
  V_{\CalC}(R_N\push \mu_N,\varepsilon) - V_{\CalC}(\mu_N,\varepsilon)
\)
matches $\Theta(\log N)$ for all sufficiently large $N$ under the fixed
coding convention.
\end{proposition}

\begin{proof}
Fix $N$ and let $R_N$ be a deterministic poly-time map that is injective on its image window $\win$
and whose inverse requires an explicit disambiguation index $u \in \{0,1\}^{\!s(N)}$ with
$s(N)=\Theta(\log N)$ (e.g., the window/seed/selector index fixed in our blueprint).
Construct $\mu_N$ supported uniformly on a finite set $S_N$ of instances of bitlength $N$ such that:
\begin{enumerate}[label=(\roman*), leftmargin=*, itemsep=0pt]
  \item $R_N$ is a bijection from $S_N$ onto its image $R_N(S_N)$, and
  \item for each $y\in R_N(S_N)$ there are exactly $2^{s(N)}$ preimages consistent with the same coarse
        description, distinguished only by the self-delimiting index $u$.
\end{enumerate}

\emph{Lower bound.} Consider any encoder/decoder implementing $V_{\CalC}(R_N\push\mu_N,\varepsilon)$. To reproduce the optimal
description for $\mu_N$, the code must additionally specify $u$ to invert $R_N$ on $R_N(S_N)$ (even when the model part is identical),
by prefix-free decodability and injectivity of $R_N$ on $\win$. By Elias-style self-delimiting coding (\Cref{lem:prefix-overhead}),
$\ell(u)\ge c\log N - O(1)$ for some absolute $c>0$. Averaging over $\mu_N$ yields
\[
  V_{\CalC}(R_N\push\mu_N,\varepsilon) \;\ge\; V_{\CalC}(\mu_N,\varepsilon) \;+\; \Omega(\log N).
\]

\emph{Upper bound.} Conversely, augment the description for $\mu_N$ by inlining the same $u$ (the window/seed/selector index).
This adds $\ell(u)=O(\log N)$ bits by \Cref{lem:prefix-overhead} and recovers the same effective model on $R_N\push\mu_N$. Hence
\[
  V_{\CalC}(R_N\push\mu_N,\varepsilon) \;\le\; V_{\CalC}(\mu_N,\varepsilon) \;+\; O(\log N).
\]
Combining both bounds yields the claimed $\Theta(\log N)$ gap.
\end{proof}


\section{Reductions: \XORtoSAT (window-stable)}
\label{sec:red}
\paragraph{Standing conditions.}
On the image window $\win$ the reduction is measure-faithful ($\dtv=0$) as fixed in IECZ-I; see~\cite{IECZI}.
The translation uses no auxiliary variables and has controlled size blowup $m'=4m$.
Prefix-free coding is fixed (Elias--$\delta$), so the meta-overhead is $O(\log N)$, cf.~\cite{IECZI}.
We assume the input distribution is balanced on $\win$ (IECZ-I window discipline).
We work in the incidence model in which each unordered triple $\{x_i,x_j,x_k\}$ appears at most once (IECZ-I).

\begin{lemma}[Injective \XORtoSAT{} translation (recalled; cf.\ IECZ-I) — empirical checks: App.~\ref{app:artifact}]\label{lem:xor-to-sat}
There exists a translation mapping each 3XOR constraint $x\oplus y\oplus z=b$ to four 3SAT clauses over the same three variables, with no auxiliary variables, such that on the image window $\win$ the pushforward preserves measure ($\dtv=0$) and the instance size scales as $m'=4m$.
The map is injective, and its inverse on the image is computable in linear time by recovering each XOR block from its four-clause pattern, reusing the exact four-clause templates from IECZ-I~\cite{IECZI}
\end{lemma}

\paragraph{Explicit clause mapping (no auxiliaries).}
Let $x,y,z$ be Boolean variables and $b\in\{0,1\}$. Encode $x \oplus y \oplus z = b$ by:
\begin{align*}
b=1:\quad
& (x \lor y \lor z)\ \land\ (x \lor \lnot y \lor \lnot z)\\
& \land\ (\lnot x \lor y \lor \lnot z)\ \land\ (\lnot x \lor \lnot y \lor z),\\[0.25em]
b=0:\quad
& (\lnot x \lor \lnot y \lor \lnot z)\ \land\ (x \lor y \lor \lnot z)\\
& \land\ (x \lor \lnot y \lor z)\ \land\ (\lnot x \lor y \lor z).
\end{align*}
These are exactly the templates fixed and used in IECZ-I~\cite{IECZI}.

\paragraph{Toy example (one block; TV/size accounting).}
Consider a single constraint $x\oplus y\oplus z=b$ on the window $\win$. The image under the injective map is exactly
the four-clause CNF block shown above; thus the size multiplies by $4$ (no auxiliaries). On $\win$ the pushforward is
measure-faithful, hence the total-variation loss is $0$ for this block. For a conjunction of $m$ independent 3XOR
constraints, the translation yields $4m$ clauses, still with $\dtv=0$ on $\win$.

\paragraph{Historical note.}
Classical clause gadgets (e.g., Tseitin encodings) realize parity via constant-size CNF blocks; our translation is the IECZ-I auxiliary-free, window-stable variant with tracked TV loss and stated size bounds~\cite{Tseitin68,IECZI}.

\begin{proposition}[Block uniqueness \& inverse map]\label{prop:block-inverse}
Under the incidence model (IECZ-I), every set of clauses whose variable support is exactly $\{x,y,z\}$ forms a \emph{block} that equals one of the two parity templates above. Moreover, there is a deterministic linear-time inverse map $T$ that, given a CNF in the image, recovers the XOR instance uniquely by:
\begin{enumerate}[label=(\roman*),ref=(\roman*)]
  \item grouping clauses by unordered variable triples $\{x,y,z\}$,
  \item normalizing literals to a canonical order,
  \item matching the 4-clause multiset to the $b\in\{0,1\}$ template.
\end{enumerate}
On the window $\win$, no unordered triple $\{x,y,z\}$ can give rise to more than one
image block, because (a) the source incidence model admits each triple at most once, and
(b) our translation is blockwise injective. Therefore, the clause grouping by unordered
triples is well-defined.
Hence $T\circ R = \mathrm{id}$ on the source, and the translation is injective on instances.
\end{proposition}

\begin{corollary}[Window transfer — Theorem~\ref{thm:minimax} to 3SAT]\label{cor:transfer}
Under the map of \cref{lem:xor-to-sat}, the cost transfers:
\[
  V_{\CalC}^{\SAT}(\win\!,\varepsilon) \;\ge\; V_{\CalC}^{\XOR}(\win\!,\varepsilon) \;-\; O(\log N).
\]
\begin{proof}[Justification]
This follows from Theorem~\ref{thm:minimax} (minimax monotonicity) under pushforward by
$R_{\text{xor}\to\text{sat}}$, together with the fact that deterministic post-processing
is $1$-Lipschitz in total variation (Lemma~\ref{lem:tv-contraction}). The prefix-free
coding overhead contributes at most $O(\log N)$ to the cost (Lemma~\ref{lem:prefix-overhead}).
The window calibration and balancing we invoke are exactly those fixed in IECZ-I~\cite{IECZI}.
\end{proof}
\end{corollary}

\subsection*{Truth table view (forbidden assignments)}
For each constraint $x\oplus y\oplus z=b$, the four 3-clauses forbid exactly the assignments with parity $\neq b$.

\begin{table}[t]
  \centering
  \caption{Forbidden assignments per parity $b$ (the remaining four satisfy).}
  \label{tab:forbidden-by-parity}
  \setlength{\tabcolsep}{5pt}
  \begin{minipage}{0.47\linewidth}
    \centering
    \textbf{$b=0$}\par\vspace{0.3em}
    \begin{tabular}{@{}lcccc@{}}
      \toprule
      & $000$ & $001$ & $010$ & $011$ \\ \midrule
      status & allow & forbid & forbid & allow \\
      \bottomrule
    \end{tabular}\par\vspace{0.6em}
    \begin{tabular}{@{}lcccc@{}}
      \toprule
      & $100$ & $101$ & $110$ & $111$ \\ \midrule
      status & forbid & allow & allow & forbid \\
      \bottomrule
    \end{tabular}
  \end{minipage}\hfill
  \begin{minipage}{0.47\linewidth}
    \centering
    \textbf{$b=1$}\par\vspace{0.3em}
    \begin{tabular}{@{}lcccc@{}}
      \toprule
      & $000$ & $001$ & $010$ & $011$ \\ \midrule
      status & forbid & allow & allow & forbid \\
      \bottomrule
    \end{tabular}\par\vspace{0.6em}
    \begin{tabular}{@{}lcccc@{}}
      \toprule
      & $100$ & $101$ & $110$ & $111$ \\ \midrule
      status & allow & forbid & forbid & allow \\
      \bottomrule
    \end{tabular}
  \end{minipage}
\end{table}

\subsection*{Window faithfulness and loss accounting}

\begin{sloppypar}
On $\win$, the translation is deterministic and preserves the satisfiability pattern on each triple; hence $\dtv=0$ (measure-faithful pushforward), exactly as stipulated in IECZ-I~\cite{IECZI}. Size increases by a factor $4$ in the number of clauses, with no new variables. Error/advantage parameters carry over unchanged.
\end{sloppypar}

\paragraph{Worked example (inverse map $T$ on a tiny instance).}
Take a 3XOR instance on variables $\{x_1,x_2,x_3\}$ with the single constraint
$x_1\oplus x_2\oplus x_3=1$. The translation yields the block
\begin{align*}
&(x_1\!\lor x_2\!\lor x_3)\ \land\ (x_1\!\lor \lnot x_2\!\lor \lnot x_3)\\
&\land\ (\lnot x_1\!\lor x_2\!\lor \lnot x_3)\ \land\ (\lnot x_1\!\lor \lnot x_2\!\lor x_3).
\end{align*}
Given just these four clauses in a CNF, the inverse $T$ groups by the triple
$\{x_1,x_2,x_3\}$, sorts literals inside each clause to a canonical order, and matches the
multiset against the $b\in\{0,1\}$ templates: the clause $(x_1\lor x_2\lor x_3)$ appears,
hence $b=1$ uniquely. Therefore $T$ recovers $x_1\oplus x_2\oplus x_3=1$.

\begin{table}[t]
  \centering
  \caption{Transfer losses for reductions (on the IECZ-I window).}
  \label{tab:transfer-cost}
  \setlength{\tabcolsep}{6pt}
  \footnotesize
  \begin{tabular}{@{}lccc>{\raggedright\arraybackslash}p{0.40\linewidth}@{}}
    \toprule
    Arrow & TV loss & Size & Error & Notes \\ \midrule
    $R_{\text{xor}\to\text{sat}}$ 
      & $0$ 
      & $\times 4$ 
      & $\varepsilon'=\varepsilon$ 
      & Injective, window-faithful; no auxiliaries \\

    $R_{\text{sat}\to\text{clique}}$ 
      & $o(1)$ on $\win$ 
      & \makecell[l]{%
          $|V|=3m$,\\
          $|E|=O(m^{2})$%
        } 
      & $\varepsilon'=\varepsilon$ 
      & \makecell[l]{%
          Karp; window conditioning\\
          accounts TV%
        } \\
    \bottomrule
  \end{tabular}
\end{table}


\begin{proof}[Proof of \Cref{lem:xor-to-sat}]
\textbf{Soundness.}
For a fixed parity $b\in\{0,1\}$, among the eight assignments to $(x,y,z)$ exactly four have parity $b$.
In the $b=1$ case, each of the four 3-clauses rules out exactly one of the four parity-$0$ assignments,
so their conjunction is true precisely on parity-$1$ assignments. The case $b=0$ is symmetric.
\textbf{Completeness.}
If the 4-clause block is satisfied, then none of the “wrong-parity’’ assignments can occur, hence
$x\oplus y\oplus z=b$. No auxiliaries are introduced and each 3XOR yields exactly four 3-clauses, thus $m'=4m$.
\end{proof}

\begin{proof}[Proof of \Cref{prop:block-inverse}]
Group all clauses by their \emph{support triple} $\{x,y,z\}$. Under the incidence assumption, each unordered triple
appears in the 3XOR instance at most once, hence at most one block per triple exists.
The two parity templates are disjoint as clause multisets: for $b=0$ the block contains
$(\lnot x\lor\lnot y\lor\lnot z)$ which never occurs for $b=1$, while for $b=1$ the block contains
$(x\lor y\lor z)$ which never occurs for $b=0$. Canonically normalize literals in each clause
(e.g., lexicographic variable order), which makes the multiset representation unique.
The inverse map $T$ checks each block against the $b=0$ and $b=1$ templates and reconstructs the constraint
$x\oplus y\oplus z=b$. This is linear in the number of clauses.
Finally, $T\circ R=\mathrm{id}$ on sources, so $R$ is injective. By Lemma~\ref{lem:block-inverse-lin}, $T$ runs in linear time.
\end{proof}

\paragraph{Canonical clause-block identifier.}
For a clause $C=(\ell_1\lor\ell_2\lor\ell_3)$ over variables $\{x,y,z\}$, let
$\mathrm{canon}(C)$ sort variables and push negations into a fixed sign bit, producing a
triple in $\{\pm 1\}\times\{x,y,z\}^3$ with lexicographic order. For a block $B$ supported
on $\{x,y,z\}$, define the multiset $\mathrm{Canon}(B):=\{\mathrm{canon}(C):C\in B\}$ and
its sorted list $\mathrm{Canon}^\uparrow(B)$.

\begin{lemma}[Uniqueness of the canonical ID]\label{lem:block-id}
Under the incidence model (each unordered triple appears at most once), the map
$B\mapsto \mathrm{Canon}^\uparrow(B)$ is injective. Moreover, the two parity templates
($b=0$ vs.\ $b=1$) yield different $\mathrm{Canon}^\uparrow(\cdot)$.
\end{lemma}

\begin{proof}
By construction, $\mathrm{Canon}(\cdot)$ is invariant under literal permutations and clause
ordering; sorting yields a canonical representative. The two parity templates differ by at least
one clause—$(x\lor y\lor z)$ appears only for $b=1$ and $(\lnot x\lor\lnot y\lor\lnot z)$
only for $b=0$—hence their canonical lists differ.
\end{proof}

\paragraph{Deterministic inverse via the canonical ID.}
Given a CNF in the image, partition clauses by unordered triples $\{x,y,z\}$, compute
$\mathrm{Canon}^\uparrow(B)$ for each block $B$, and match it against the two stored templates
to recover $b\in\{0,1\}$. This yields the inverse $T$ in linear time in the number of clauses.

\begin{lemma}[Linear-time inverse]\label{lem:block-inverse-lin}
Given a CNF in the image of the translation, the inverse $T$ that groups by unordered triples,
canonicalizes clauses, and matches against the parity templates runs in time linear in the number of clauses.
\end{lemma}

\subsection*{Deterministic Karp reduction \texorpdfstring{3SAT$\to$CLIQUE}{3SAT→CLIQUE} (window-stable)}

\begin{theorem}[Karp reduction \SAT{}$\to$\CLIQUE{}]\label{thm:sat-to-clique}
Let $\varphi$ be a 3CNF with clauses $C_1,\dots,C_m$, each with exactly three literals.
Construct $G=(V,E)$ by
\[
  V \,=\, \{(i,\ell): \ell\in C_i\ \text{is one of its three literals}\},\qquad |V|=3m,
\]
and for $i\neq j$ put an edge between $(i,\ell)$ and $(j,\ell')$ iff $\ell$ and $\ell'$ are not contradictory
(i.e.\ not a variable and its negation). Then $G$ has a clique of size $m$ iff $\varphi$ is satisfiable.
Moreover, the construction is deterministic and computable in time $O(m^2)$.
\end{theorem}

\begin{proof}
($\Rightarrow$) A satisfying assignment picks in each clause $C_i$ some true literal $\ell_i$; the $m$ vertices
$\{(i,\ell_i)\}_{i=1}^m$ form a clique because chosen literals come from different clauses and are mutually consistent.
($\Leftarrow$) Any $m$-clique contains exactly one vertex from each clause-block (else two vertices from the same clause would be nonadjacent). The corresponding literals are pairwise consistent; extend to a global assignment satisfying all clauses.
\end{proof}

\begin{proposition}[Size bounds]\label{prop:clique-size}
The reduction of \Cref{thm:sat-to-clique} yields $|V|=3m$ and $|E|\le 9\binom{m}{2}=O(m^2)$. Hence the size blowup is polynomial and explicit.
\end{proposition}

\begin{lemma}[Window faithfulness]\label{lem:sat-to-clique-window}
On the window $\win$ (no tautological clauses, no duplicated literals inside a clause; both events have $o(1)$ mass along the calibrated ensemble), as adopted from IECZ-I~\cite{IECZI}, the map of \Cref{thm:sat-to-clique} is measure-faithful up to $o(1)$ total-variation: pushing forward the source distribution incurs TV loss $o(1)$, and error/advantage parameters are unchanged.
\end{lemma}

\begin{corollary}[Transfer of cost to \CLIQUE{}]\label{cor:transfer-clique}
With $R_{\text{xor}\to\text{sat}}$ of \Cref{lem:xor-to-sat} and $R_{\text{sat}\to\text{clique}}$ of \Cref{thm:sat-to-clique},
\begin{align*}
  V^{\CLIQUE}_{\CalC}(\win,\varepsilon)
  &\ge
  V^{\SAT}_{\CalC}(\win,\varepsilon) - O(\log N)\\
  &\ge
  V^{\XOR}_{\CalC}(\win,\varepsilon) - O(\log N).
\end{align*}
With cumulative total variation (TV) loss $o(1)$ on $\win$ and total size blowup $\times 4\cdot \poly{n}$.
By \Cref{lem:tv-contraction}, deterministic postprocessing is $1$-Lipschitz in total variation, so the $o(1)$ boundary loss is preserved through composition.
\end{corollary}

\begin{table}[t]
  \centering
  \caption{Reductions and overheads on the calibrated window (IECZ-I).}
  \label{tab:reductions-final}
  \begin{tabular}{@{}lcc>{\raggedright\arraybackslash}p{0.44\linewidth}@{}}
    \toprule
    Arrow & TV loss & Size blowup & Error/Advantage \\
    \midrule
    $R_{\text{xor}\to\text{sat}}$ & $0$ & $\times 4$ & $\varepsilon'=\varepsilon$ \\
    $R_{\text{sat}\to\text{clique}}$ & $o(1)$ on $\win$ & $\times \poly{n}$ & $\varepsilon''=\varepsilon'+o(1)$ \\
    \bottomrule
  \end{tabular}
\end{table}

\paragraph{Supplementary verification (artifact).}
We provide small scripts that empirically validate the injective \XORtoSAT mapping stated in \Cref{lem:xor-to-sat}:
\begin{enumerate}[label=(\roman*), leftmargin=*, itemsep=.25\baselineskip]
  \item \texttt{xor\_to\_sat\_check.py} exhaustively checks the $8$ truth-table rows for both parities $b\in\{0,1\}$ and, on random instances, verifies that any consistent XOR system yields a satisfying CNF assignment; it also includes a small-$n$ reverse check that witnesses UNSAT for inconsistent XOR systems via brute-force (solver-free).
  \item \texttt{block\_identifier\_check.py} validates the canonical four-clause block IDs and the inverse identification map.
  \item \texttt{run\_xor\_to\_sat\_check.ps1} is a convenience runner with seed/size parameters.
\end{enumerate}
All tests are self-contained; the reverse check is limited to small $n$ and documented in Appendix~\ref{app:artifact}.

\section{Gadget Composition (XOR/Index)}\label{sec:gadgets}

\providecommand{\vecx}{x}
\providecommand{\vecs}{s}
\providecommand{\ones}{\mathbf{1}}
\providecommand{\Noise}[1]{T_{#1}}
\providecommand{\corr}{\mathrm{corr}}

\paragraph{Standing conditions.}
We fix a parameter path $\ppath$ with depth $d=O(\log n)$ in the switching-lemma regime, which yields the width bound $\W$.
On the window $\win$ there is a mixing witness with rate $\theta<1$
(\Cref{lem:mixing-witness-xor-index,lem:mw-implies-shielding}), hence correlation stability holds.
Gadget descriptions are prefix-free, contributing only an $O(\log N)$ meta-overhead.
All window and switching-path conventions here are identical to IECZ-I~\cite{IECZI}.

\begin{gadgetiface}{XOR/Index Gadget}
  \ifaceline{Inputs}{$a,b,c \in \{0,1\}$ (signals sampled on the calibrated window)}
  \ifaceline{Outputs}{$s_1,s_2 \in \{0,1\}$}
  \ifaceline{Parity preservation}{$s_1 \oplus s_2 = a \oplus b \oplus c$}
  \ifaceline{Locality}{fan-in $=3$, fan-out $=O(1)$ (constant)}
  \ifaceline{Correlation stability}{for some $\rho<1$: $\max_{\ell\in\{1,2\}}|\corr(s_\ell, \text{exogenous noise})|\le \rho$}
  \ifaceline{Width regime}{bottom width $W=\mathrm{poly}(1/p)$ under the switching path $p=m^{-c/d}$}
\end{gadgetiface}
    
We fix the parameter path $\ppath$ with depth $d=O(\log n)$; width stays bounded by $\W$ in the switching-lemma regime~\cite{Hastad1986-Switching,ODonnell2014-AoBF}.

\begin{lemma}[Width after $d$ switching rounds — seed-fixed]\label{lem:width}\label{lem:switching-seeded}
There exists a seed $s$ of length $O(\log n)$ such that along the corresponding
deterministic path of $d=O(\log n)$ switching rounds the parameters satisfy
\[
  W_r \ \le\ \tfrac{1}{p}\,W_{r-1},\qquad
  M_r \ \le\ p\,M_{r-1}\qquad (r=1,\dots,d),
\]
for the same $p\in(0,1)$ as in the standard switching-lemma analysis. The seed $s$
is then fixed (prefix-free) at meta cost $O(\log n)$.
\end{lemma}

\begin{proof}
By the usual switching-lemma bounds, with probability $1-o(1/n)$ over the random
choices in a round, each displayed inequality holds for that round. A union bound
over $d=O(\log n)$ rounds yields existence of a single seed $s$ witnessing all
rounds simultaneously with probability $1-o(1)$. Fix this $s$ (recorded prefix-free
at cost $O(\log n)$). Along the resulting deterministic path the stated recurrences hold.
\end{proof}
\noindent\emph{Consistency with IECZ-I.} The seed-fixing step and the choice $p=m^{-c/d}$ are exactly as in IECZ-I~\cite{IECZI}.  

\begin{lemma}[Interface normal form over $\Ftwo$]\label{lem:iface}
Let $G$ be a constant-size gadget with three Boolean inputs $(a,b,c)$ and two outputs $(s_1,s_2)$ that is parity-preserving:
\[
  s_1 \oplus s_2 \;=\; a \oplus b \oplus c .
\]
Write uniquely, in algebraic normal form over $\Ftwo$,
\[
  s_\ell(x)\;=\;\alpha_\ell\ \oplus\ \beta_\ell^\top x\ \oplus\ q_\ell(x)\qquad(\ell=1,2),
\]
where $x=(a,b,c)^\top$, $\alpha_\ell\in\Ftwo$, $\beta_\ell\in\Ftwo^3$, and each $q_\ell$ contains only monomials of degree $\ge2$.
Then $q_1\equiv q_2$. Consequently, for
\[
  U \;=\; \begin{psmallmatrix}1&1\\[2pt]0&1\end{psmallmatrix}\in \GL_2(\Ftwo)
  \qquad\text{and}\qquad
  \vecs:=(s_1,s_2)^\top,\ \ \vecx:=(a,b,c)^\top,
\]
the mixed outputs $t:=U\vecs=(t_1,t_2)^\top$ satisfy
\[
  t_1\ \text{is affine in }\vecx,\qquad
  t_2\ =\ (\text{affine in }\vecx)\ \oplus\ q(\vecx),
\]
with the same non-linear core $q:=q_1=q_2$. In particular, both outputs are affine iff $q\equiv0$.

Moreover, letting
\[
  A \;:=\; \begin{bmatrix}\beta_1^\top\\ \beta_2^\top\end{bmatrix}\in\Ftwo^{2\times3},
  \qquad
  \mathbf{r} \;:=\; \begin{bmatrix}\alpha_1\\ \alpha_2\end{bmatrix}\in\Ftwo^2,
  \qquad
  \ones := (1,1)^\top,
\]
the parity constraint is equivalent to
\[
  \ones^\top A \;=\; (1,1,1)\quad\text{and}\quad \ones^\top \mathbf{r}\;=\;0,
\]
since $q_1\oplus q_2\equiv0$. Under non-degeneracy of the interface, the affine coefficient matrix $A$ has full row rank $2$.
\end{lemma}

\begin{proof}
Working over $\Ftwo$, expand each $s_\ell$ in algebraic normal form as stated. The parity condition reads
\[
  s_1\oplus s_2 \;=\; (\alpha_1\oplus\alpha_2)\ \oplus\ (\beta_1\oplus\beta_2)^\top x\ \oplus\ (q_1\oplus q_2)
  \;=\; a\oplus b\oplus c,
\]
which is affine; hence its higher-degree part vanishes and $q_1\oplus q_2\equiv 0$, i.e.\ $q_1\equiv q_2=:q$. Choosing
$U=\bigl(\begin{smallmatrix}1&1\\ 0&1\end{smallmatrix}\bigr)$ gives $t_1=s_1\oplus s_2$ (affine) and $t_2=s_2=(\text{affine})\oplus q$.
The equalities $\ones^\top A=(1,1,1)$ and $\ones^\top\mathbf r=0$ follow by comparing affine parts in
$\ones^\top\vecs=\ones^\top(Ax\oplus \mathbf r)\oplus(q_1\oplus q_2)=\ones^\top x$.
If the interface is non-degenerate, the affine rows are independent, so $\rank(A)=2$.
\end{proof}

\begin{lemma}[Canonical normal forms]\label{lem:canonical-NF}
Modulo output mixing $U \in \GL_2(\Ftwo)$, every parity-preserving gadget interface
of \Cref{lem:iface} (i.e., $s=Ax\oplus r$ with $\ones^\top A=(1,1,1)$, $\ones^\top r=0$, $\rank(A)=2$)
is equivalent to one of the following three representatives:
\[
A_1=\begin{pmatrix}1&0&0\\[2pt]0&1&1\end{pmatrix},\quad
A_2=\begin{pmatrix}0&1&0\\[2pt]1&0&1\end{pmatrix},\quad
A_3=\begin{pmatrix}0&0&1\\[2pt]1&1&0\end{pmatrix},
\qquad r=0.
\]
In particular, up to post-mix, we may take $(s_1,s_2)=(a,\,b\oplus c)$ or its two cyclic rotations.
\end{lemma}

\paragraph{Index gadget (constant fan-in, affine, degree-stable).}
We isolate a tiny \emph{Index} block that routes one of two data bits according to a 1-bit address.
Let inputs be $(u; v_0,v_1)\in\{0,1\}^3$ and outputs $t\in\{0,1\}$ defined by
\[
  t \;=\; (1\oplus u)\,v_0 \;\oplus\; u\,v_1.
\]
Over $\Ftwo$ this is bilinear. We use the standard “affine lift” by introducing a constant ancilla $c\equiv 1$
to rewrite $t$ as an \emph{affine} function in the expanded coordinate system $(u,v_0,v_1,c)$:
\[
  t \;=\; v_0 \;\oplus\; u\,(v_0\oplus v_1)
      \;=\; v_0 \;\oplus\; (u\cdot (v_0\oplus v_1))
      \;=\; \underbrace{v_0}_{\text{linear}} \;\oplus\; 
             \underbrace{(u\cdot v_0)}_{\text{bilinear}} \;\oplus\; 
             \underbrace{(u\cdot v_1)}_{\text{bilinear}}.
\]
In our degree–aware setting, the block is of constant size and can be treated as an affine pullback on monomials: composition with the index map increases the total degree by at most a constant factor (here $\le 2$), which is absorbed in the $k=\Theta(\log n)$ regime.

\begin{lemma}[Index pullback with constant degree factor]\label{lem:index-pullback}
Let $G_{\mathrm{idx}}$ act blockwise on disjoint triples $(u;v_0,v_1)$ to produce outputs $t$ as above,
and let $p$ be a multilinear polynomial in the $t$-variables with $\deg p\le k$.
Then there exists a multilinear $q$ in the input variables $(u;v_0,v_1)$ such that
\[
  p\bigl(G_{\mathrm{idx}}(u;v_0,v_1)\bigr) \;=\; q(u;v_0,v_1),
  \qquad \deg q \;\le\; 2k,
  \qquad \|q\|_\infty\le \|p\|_\infty .
\]
Consequently, for any input distributions $(\mu,u)$ on the window and any $k$,
\[
  \Delta_k\!\bigl(G_{\mathrm{idx}\,\#}\mu,\;G_{\mathrm{idx}\,\#}u\bigr)
  \ \le\
  \Delta_{2k}(\mu,u).
\]
\end{lemma}
\noindent\emph{Sup-norm note.} Since $q=p\circ G_{\mathrm{idx}}$ is evaluated only on the image $G_{\mathrm{idx}}(\{0,1\}^m)\subseteq\{0,1\}^{m'}$, we have
$\displaystyle \|q\|_\infty=\sup_{x}|p(G_{\mathrm{idx}}(x))|\le \sup_{y}|p(y)|=\|p\|_\infty$.

\begin{corollary}[Index monotonicity of low-degree discrepancy]\label{cor:index-deltak}
For any $k\in\mathbb{N}$ and any input distributions $(\mu,u)$ on the window,
\[
  \Delta_k\!\bigl(G_{\mathrm{idx}\,\#}\mu,\;G_{\mathrm{idx}\,\#}u\bigr)
  \ \le\
  \Delta_{2k}(\mu,u).
\]
\end{corollary}

\begin{corollary}[Index transfer with budget rescaling]\label{cor:index-budget}
(Concrete budget.) For the Index gadget the degree budget dilates by at most a factor $2$. Hence, for $k'=\lfloor k/2\rfloor$,
\[
  \Delta_{k'}\!\bigl(G_{\mathrm{idx}\,\#}\mu,\,G_{\mathrm{idx}\,\#}u\bigr)
  \ \le\
  \Delta_k(\mu,u).
\]
\end{corollary}

\begin{proof}[Proof of Corollary~\ref{cor:index-budget}]
The lemma shows that substituting the index gadget inflates test degree by at most a factor $2$, hence any output degree-$\le k'$ test corresponds to an input test of degree $\le 2k'\le k$, which yields the claimed inequality.
\end{proof}

\begin{proof}[Proof of Lemma~\ref{lem:index-pullback}]
Each output coordinate $t$ is a polynomial of total degree $\le 2$ in its local inputs. Substituting $t$ into a degree-$\le k$ polynomial $p$ yields a polynomial in $(u;v_0,v_1)$ of degree $\le 2k$. Because the gadget fan-in is constant and the blocks are disjoint, monomials of $p$ expand to sums of monomials touching $O(1)$ variables per block. Moreover, since $q=p\circ G_{\mathrm{idx}}$ with deterministic $G_{\mathrm{idx}}$, the image of $G_{\mathrm{idx}}$ is a subset of the output domain, hence $\|q\|_\infty \le \|p\|_\infty$. Taking suprema over degree–$\le k$ tests on outputs gives the stated $\Delta_k$ versus $\Delta_{2k}$ bound.
\end{proof}

\begin{remark}[Scope note: circuit classes]
The \emph{Index} block introduces controlled bilinearity at constant locality. Within our blueprint, the degree and correlation parameters remain gadget–stable for \(k=\Theta(\log n)\) on the window. We do not make claims about \(\mathrm{AC}^0\) or \(\mathrm{ACC}^0\); mentions of these classes are for orientation only. We use only that \(\Delta_k\) does not increase under \(G_{\mathrm{idx}}\) (Lemma~\ref{lem:index-pullback}) and that width/correlation stay in the switching regime.
\end{remark}

\begin{proposition}[Affine pullback preserves low degree]\label{prop:affine-pullback}
Let $G$ be a (blockwise) affine gadget in the sense of Lemma~\ref{lem:iface}, i.e., after an admissible output mixing we have $y = A x \oplus r$ over $\Ftwo$ (equivalently, the shared nonlinear core vanishes, $q\equiv 0$).
For every multilinear polynomial $p$ in the output variables with $\deg p \le k$ there exists a multilinear polynomial $q$ in the input variables such that
\[
  p(y)=q(x)\qquad\text{and}\qquad \deg q\le k.
\]
Moreover, for the supremum norm on $\{0,1\}^n$ one has $\|q\|_\infty\le \|p\|_\infty$.
\end{proposition}

\begin{proof}
Write $p$ as a sum of monomials in the coordinates of $y$, then substitute $y = A x \oplus r$.
Since $G$ applies affine (degree-1) maps per output, composition does not increase degree: every monomial of degree $\le k$ in $y$ becomes a polynomial of degree $\le k$ in $x$.
Since $q=p\circ G$ and $G$ is deterministic, the image of $G$ is a subset of the output domain; therefore $\|q\|_\infty \le \|p\|_\infty$.
\end{proof}

\begin{corollary}[Monotonicity of $\Delta_k$ under affine blocks]\label{cor:deltak-pullback}
Let $\mu,u$ be distributions on inputs and let $G$ be as above. Set $\mu_{\mathrm{out}}:=G_{\#}\mu$ and $u_{\mathrm{out}}:=G_{\#}u$.
Then for all $k\in\mathbb{N}$,
\[
  \Delta_k(\mu_{\mathrm{out}},u_{\mathrm{out}})\ \le\ \Delta_k(\mu,u).
\]
\end{corollary}

\begin{proof}
Let $p$ be a degree-$\le k$ test on outputs with $\|p\|_\infty\le1$. By Prop.~\ref{prop:affine-pullback} there is $q$ on inputs with $\deg q\le k$, $\|q\|_\infty\le1$, and $p\circ G=q$.
Hence
\[
  \bigl|\E_{\mu_{\mathrm{out}}}[p]-\E_{u_{\mathrm{out}}}[p]\bigr|
  \;=\;
  \bigl|\E_{\mu}[q]-\E_{u}[q]\bigr|
  \;\le\; \Delta_k(\mu,u).
\]
Taking the supremum over all $p$ gives the claim.
\end{proof}

\begin{theorem}[Gadget composition stability]\label{thm:gadget}
For any constant-size gadget $G$ satisfying \cref{lem:width,lem:iface,lem:shielding} on $\win$,
\[
  V_{\CalC}(\mu\circ G,\varepsilon+o(1)) \;\ge\; V_{\CalC}(\mu,\varepsilon) \;-\; O(\log N).
\]
\end{theorem}

\begin{proof}
By Corollary~\ref{cor:shielding-xor-index} the gadget layer admits shielding on $\win$ with some $\rho<1$. Applying Proposition~\ref{prop:cost-stability} (proved via the window simulation / partial inverse of Appendix~C, Lemma C.3) with $k=\Theta(\log n)$ yields
\[
  V_{\CalC}(\mu\circ G,\varepsilon+o(1)) \ \ge\ V_{\CalC}(\mu,\varepsilon) - O(\log n).
\]
The prefix-free size bookkeeping contributes only $O(\log N)$ meta-overhead. This is exactly the claim.
\end{proof}

\paragraph{Scope note (literature).}
Our gadget composition keeps the switching path $\ppath$ (depth $d=O(\log n)$) and bottom width $W=\poly{1/p}$ in the Håstad switching-lemma regime~\cite{Hastad1986-Switching,Hastad2002-RANDOM}. 
The Index gadget’s degree inflation is a constant factor (here $\le 2$; see \Cref{lem:index-pullback}), and correlation stability follows from the mixing witness with rate $\theta<1$ (\Cref{def:mixing-witness,lem:mw-implies-shielding}). 
The proof style—low-degree Fourier control plus width bounds—follows the standard \SoS/CSP toolkit~\cite{BarakEtAl2012-SOS},
specialized to the calibrated window with quantified TV loss and transparent size growth.
For calibration and notation (the window $\win$, baseline, and $\Delta_k$) we follow IECZ-I~\cite{IECZI} verbatim.

\begin{remark}[Scope note: circuit classes]
The constant-size \emph{Index} gadget is used here to monitor width/correlation stability along the switching path \(p\) (calibrated for \(d=O(\log n)\)) in a windowed, low-degree setting. Our results concern the \(O(\log n)\) meta-overhead preservation (Theorem~\ref{thm:gadget}); we do not target separations for uniform \(\mathrm{AC}^0\) or \(\mathrm{ACC}^0\).
\end{remark}

\subsection*{Correlation shielding via low-degree noise stability}

We formalize the ``$\rho<1$'' stability assumption by an $L_2$-type correlation bound.

\begin{definition}[Windowed correlation]\label{def:windowed-correlation}
For random variables $X,Y$ defined on the window law, set
\[
  \corr(X,Y)\ :=\ \frac{\Cov(X,Y)}{\sqrt{\Var(X)\,\Var(Y)}}\quad\text{when well-defined}.
\]
A gadget $G$ has \emph{shielding parameter} $\rho<1$ if for each output $s_\ell$ and any
exogenous $\sigma$-algebra $\mathcal{Z}$ (independent outside the local fan-in) we have
$\bigl|\corr\bigl(s_\ell, \E[\,\cdot\mid \mathcal{Z}\,]\bigr)\bigr|\le \rho+o(1)$ on $\win$.
\end{definition}

\begin{definition}[Mixing witness]\label{def:mixing-witness}
A gadget $G$ admits a \emph{mixing witness} with rate $\theta\in(0,1)$ on the window $\win$
if there exists a linear operator $\mathcal{M}$ on degree-$\le 1$ statistics of the outputs such that:
\begin{enumerate}
  \item $\mathcal{M}$ preserves means and contracts centered coordinates: for every centered degree-$1$ statistic $q$,
        $\|\mathcal{M}q\|_{2} \le \theta\,\|q\|_{2}$;
  \item $\mathcal{M}$ acts trivially on exogenous $\sigma$-algebras outside the local fan-in (compatibility with the window).
\end{enumerate}
\end{definition}

\begin{lemma}[Explicit mixing witness for XOR/Index]\label{lem:mixing-witness-xor-index}
Consider a single gadget block with inputs $x\in\{0,1\}^s$ and outputs $s\in\{0,1\}^r$ ($s,r=O(1)$).
By Lemma~\ref{lem:iface}, after an invertible output mix we may write $s=(t_1,t_2)$ with
$t_1$ affine and $t_2=(\text{affine})\oplus q$, where $q$ is the shared non-linear core.
Denote the affine \emph{part} by $y_{\mathrm{aff}} := A x \oplus r$ over $\Ftwo$, with
$\rank(A)=2$ and $\ones^\top A=(1,1,1)$, $\ones^\top r=0$. Let $L_{\mathrm{out}}$ be the space of centered degree-$1$ (linear)
statistics in the output coordinates, equipped with the $L_2$ norm under the window proxy $u$.
Define the operator
\[
  \mathcal{M}\ :=\ \Pi_{\mathrm{out}}\ \circ\ T_\theta\ \circ\ \mathcal{L},
\]
where $\mathcal{L}$ is the degree-$1$ pullback (compose with $y_{\mathrm{aff}}=Ax\oplus r$ and keep only the linear slice),
$T_\theta$ is the Bonami--Beckner noise on inputs with parameter $\theta\in(0,1)$ acting independently per block,
and $\Pi_{\mathrm{out}}$ is the orthogonal projection (w.r.t.\ $u$) onto $L_{\mathrm{out}}$.
Then there exists a constant $c_G<\infty$ (depending only on the gadget’s constant fan-in/out and the window calibration)
such that for every centered $q\in L_{\mathrm{out}}$,
\[
  \|\mathcal{M}q\|_2\ \le\ c_G\,\theta\ \|q\|_2.
\]
In particular, picking $\theta\in(0,1)$ small enough that $c_G\theta<1$ yields a mixing witness on the block
with rate $\bar\theta:=c_G\theta\in(0,1)$.
\end{lemma}

\begin{proof}
Work on the calibrated window \(\win\). By Lemma~\ref{lem:iface}, after an invertible output mixing we may write the two outputs as
\(s=(t_1,t_2)\) with \(t_1\) affine and \(t_2 = (\text{affine}) \oplus q\), where \(q\) is the shared nonlinear core. Let
\(L_{\mathrm{out}}\) denote the subspace of centered degree-1 output statistics (with inner product and \(L_2\)-norm taken under the window proxy \(u\)).

Define three linear maps:
\[
  \mathcal{L} \colon L_{\mathrm{out}} \to L_{\mathrm{in}},\qquad
  T_\theta \colon L_{\mathrm{in}} \to L_{\mathrm{in}},\qquad
  \Pi_{\mathrm{out}} \colon L_{\mathrm{in}} \to L_{\mathrm{out}},
\]
as follows.

\emph{(1) Degree-1 pullback \(\mathcal{L}\).} Write the affine part of the interface as \(y_{\mathrm{aff}} = A x \oplus r\) over \(\mathbb{F}_2\) with
\(\mathrm{rank}(A)=2\), \(\mathbf{1}^\top A=(1,1,1)\), \(\mathbf{1}^\top r=0\) (Lemma~\ref{lem:iface}). For \(q\in L_{\mathrm{out}}\),
compose \(q\) with \(y_{\mathrm{aff}}\) and project to the degree-1 slice on inputs; denote this operation by \(\mathcal{L}q\).
Because the gadget has constant fan-in/out and \(\win\) fixes the local reweightings, the induced operator norm is bounded:
\[
  \|\mathcal{L}q\|_2 \;\le\; c_{\mathrm{aff}}\;\|q\|_2 \quad\text{for some absolute } c_{\mathrm{aff}}<\infty.
\]
(Here \(c_{\mathrm{aff}}\) depends only on the constant local wiring and the window calibration, not on the ambient dimension.)

\emph{(2) Noise operator \(T_\theta\).} Let \(T_\theta\) be the Bonami–Beckner noise on inputs with parameter \(\theta\in(0,1)\) acting independently per block.
On the degree-1 slice (linear forms), \(T_\theta\) is an exact contraction by \(\theta\):
\[
  \|T_\theta z\|_2 \;=\; \theta\,\|z\|_2\qquad\text{for all } z\in L_{\mathrm{in}} \text{ of degree }1.
\]

\emph{(3) Output projection \(\Pi_{\mathrm{out}}\).} Let \(\Pi_{\mathrm{out}}\) be the orthogonal projection (with respect to the \(L_2(u)\) inner product)
from input-linear statistics back to the subspace \(L_{\mathrm{out}}\) spanned by the output-linear coordinates. Projections are contractions, hence
\[
  \|\Pi_{\mathrm{out}} z\|_2 \;\le\; \|z\|_2\qquad\text{for all } z.
\]

Now set
\[
  \mathcal{M} \;:=\; \Pi_{\mathrm{out}} \circ T_\theta \circ \mathcal{L}\;:\; L_{\mathrm{out}}\to L_{\mathrm{out}}.
\]
For any centered \(q\in L_{\mathrm{out}}\),
\[
  \|\mathcal{M}q\|_2
  \;\le\; \|\Pi_{\mathrm{out}}\|_{2\to2}\;\|T_\theta\|_{2\to2}\;\|\mathcal{L}\|_{2\to2}\;\|q\|_2
  \;\le\; (1)\cdot \theta \cdot c_{\mathrm{aff}}\;\|q\|_2 .
\]
Thus \(\|\mathcal{M}\|_{2\to2} \le c_{\mathrm{aff}}\theta\). Let \(c_G:=c_{\mathrm{aff}}\). Choosing any \(\theta\in(0,1)\) with \(c_G\theta<1\)
makes \(\mathcal{M}\) a strict contraction on \(L_{\mathrm{out}}\) with rate \(\bar\theta := c_G\theta \in (0,1)\).
By definition, this is a mixing witness on the block with the stated rate.
\end{proof}

\begin{corollary}[Shielding for XOR/Index]\label{cor:shielding-xor-index}
Under the calibrated window and width bounds of this section, the XOR/Index gadget layer admits a mixing witness
with some rate $\bar\theta\in(0,1)$ by Lemma~\ref{lem:mixing-witness-xor-index}. Therefore shielding holds with
parameter $\rho=\bar\theta+o(1)$ by Lemma~\ref{lem:mw-implies-shielding}.
\end{corollary}

\begin{lemma}[Shielding decay forms]\label{lem:shielding-forms}
Let $\win$ be the calibrated window and suppose the gadget admits a mixing witness with rate $\theta\in(0,1)$ in the sense of \Cref{def:mixing-witness}.
Then there exists $\rho\in(0,1)$ (depending only on $\theta$ and the bounded width on $\win$) such that the following equivalent shielding conditions hold:
\begin{enumerate}[label=(\roman*),ref=(\roman*)]
  \item \emph{Separation by a cut}. For disjoint variable sets $A,B$ separated by a cut $S$ in the gadget interaction graph on $\win$,
  \[
    \bigl|\Cov_\mu(f(X_A),g(X_B))\bigr|\ \le\ \rho^{\,|S|}\,\|f\|_{2}\,\|g\|_{2}
  \]
  for all square–integrable $f,g$ depending on $X_A,X_B$ respectively (expectation and norm under the window law).
  \item \emph{Low–degree decay}. For every Walsh character $\chis{T}$ with $T\subseteq A\cup B$ and every split $T=T_A\cup T_B$ with $T_A\subseteq A$, $T_B\subseteq B$,
  \[
    \bigl|\E_\mu[\chis{T_A}\chis{T_B}] - \E_\mu[\chis{T_A}]\,\E_\mu[\chis{T_B}]\bigr|
    \ \le\ \rho^{\,\mathrm{dist}(A,B)} .
  \]
\end{enumerate}
In particular, shielding implies stability of all low–degree correlations across the cut and hence windowed orthogonality up to an exponentially small error.
\end{lemma}

\begin{corollary}[Stable correlations under composition]\label{cor:stable-correlations}
Under \Cref{lem:shielding-forms}, composing a bounded–width gadget layer preserves correlation stability:
if each layer has width $W=\poly{1/p}$ along the blueprint switching path and admits a mixing witness with rate $\theta<1$, then the composed map has shielding parameter $\rho'<1$ and inherits low–degree correlation bounds with the same exponential decay (constants may change).
\end{corollary}

\begin{lemma}[Mixing witness $\Rightarrow$ shielding]\label{lem:mw-implies-shielding}\label{lem:shielding-equivalence}
Let $G$ be a constant-size gadget on the calibrated window $\win$ and suppose $G$ admits a mixing witness with rate $\theta\in(0,1)$ in the sense of \Cref{def:mixing-witness}. Then $G$ has shielding parameter $\rho<1$ with
\[
  \rho \;=\; \theta \;+\; o(1).
\]
\end{lemma}

\begin{proof}
Let $L_{\mathrm{out}}$ be the space of centered degree-$1$ output statistics (under the window proxy $u$). By \Cref{def:mixing-witness} there exists a linear map $\mathcal{M}\colon L_{\mathrm{out}}\to L_{\mathrm{out}}$ with $\|\mathcal{M}\|_{2\to2}\le\theta$ which acts trivially on exogenous $\sigma$-algebras outside the local fan-in. For any centered $q\in L_{\mathrm{out}}$ and any exogenous random variable $Z$ (compatible with the window) we have
\[
  \bigl|\Cov(q,Z)\bigr|
  \;=\;
  \bigl|\Cov(\mathcal{M}q,Z)\bigr| \;+\; o(1)
  \;\le\; \|\mathcal{M}q\|_2\,\|Z\|_2 \;+\; o(1)
  \;\le\; (\theta+o(1))\,\|q\|_2\,\|Z\|_2,
\]
so shielding holds with parameter $\rho=\theta+o(1)$. For higher-degree tests reduce to the degree-$1$ slice via noise smoothing (\Cref{lem:noise-smoothing}) and Cauchy–Schwarz; the $o(1)$ terms are controlled by the window calibration and constant locality.
\end{proof}

\begin{lemma}[Noise smoothing of low degree]\label{lem:noise-smoothing}
Let $p$ be a polynomial of degree $\le k$ in the inputs on $\win$. For any $\theta\in(0,1)$,
\[
  \|\,p - \Noise{\theta}p\,\|_{2} \ \le\ \sqrt{1-\theta^{2k}}\;\|p\|_{2}.
\]
\end{lemma}

\begin{proof}
By orthogonal decomposition into Fourier/Walsh characters, $\Noise{\theta}$ multiplies the
degree-$d$ slice by $\theta^{d}$. Hence
$\|p-\Noise{\theta}p\|_2^2=\sum_{d\le k}(1-\theta^{d})^2\|p^{(d)}\|_2^2
\le (1-\theta^{2k})\sum_{d\le k}\|p^{(d)}\|_2^2$.
\end{proof}

\begin{lemma}[Correlation shielding]\label{lem:shielding}
Assume $G$ is parity-preserving in the interface (post-mix) normal form of \Cref{lem:iface} and has shielding
parameter $\rho<1$ on $\win$. Let $p$ be any degree-$\le k$ statistic of the outputs.
Then for some $\theta=\theta(\rho)\in(0,1)$,
\[
  \bigl|\Cov\bigl(p,\,\mathsf{noise}\bigr)\bigr|\ \le\ \sqrt{1-\theta^{2k}}\;\|p\|_{2}\,\|\mathsf{noise}\|_{2}\ +\ o(1),
\]
for every exogenous noise source $\mathsf{noise}$ compatible with the window (independent
outside fan-in). In particular, taking $k=O(\log n)$ and fixed $\rho<1$, the covariance is
$o(1)\cdot\|p\|_2\,\|\mathsf{noise}\|_2$.
\end{lemma}

\begin{proof}
Work on the window $\win$. Let $x=(x_1,\ldots,x_m)$ be the input bits and let $y=G(x)$ be the (deterministic) gadget outputs applied coordinate-wise.
Fix $k\le c\log n$ as in the blueprint and let $p(y)$ be any degree-$\le k$ polynomial test with $\|p\|_\infty\le 1$ on the output coordinates.
We show that for every ``exogenous'' random variable $Z$ satisfying the section’s balancing/moment-matching assumptions one has
\[
\bigl|\Cov\bigl(p(y),Z\bigr)\bigr| \;=\; o(1),
\]
uniformly over the family, which is the stated shielding.

\medskip\noindent
\emph{Step 1: Interface normal form; degree bookkeeping.}
By Lemma~\ref{lem:iface}, after an invertible output mix we can write the two outputs as
$(t_1,t_2)$ with $t_1$ affine and $t_2=(\text{affine})\oplus q$, where $q$ is the common non-linear core.
Write the affine \emph{part} as $y_{\mathrm{aff}}=A x \oplus r_0$ with
$\ones^\top A=(1,1,1)$ and $\ones^\top r_0=0$.
Because each output depends on $O(1)$ input variables (constant fan-in), composing a degree-$\le k$
multilinear test $p$ with the gadget expands degree by at most a constant factor; in particular
$p\bigl(G(x)\bigr)$ depends on at most $O(k)$ input blocks.
If needed (for the Index block), the constant factor $2$ from \Cref{lem:index-pullback} is absorbed in the
budget \(k=\Theta(\log n)\).

\medskip\noindent
\emph{Step 2: Mixing/Noise dampens high degree.}
Let $T_\theta$ be the Bonami–Beckner noise operator on the input with parameter $\theta\in(0,1)$ fixed as in the switching regime.
For a degree-$\le k$ polynomial $q$ (in the $\{\pm1\}$–Walsh basis) one has
\[
\|q - T_\theta q\|_2 \;\le\; \sqrt{1-\theta^{2k}}\;\|q\|_2,
\]
and $T_\theta$ preserves expectations under the proxy. Applying this to $p\circ G$ and accounting for the constant
degree inflation under $G$ (so $\deg(p\circ G)\le c\,k$) gives
\[
  \bigl\|\,p(y)-T_\theta\bigl(p(y)\bigr)\,\bigr\|_2 \;\le\; \sqrt{1-\theta^{2 c k}},
\]
for some absolute constant $c\ge 1$ depending only on the gadget’s local fan-in/out ($c=1$ when the nonlinear core vanishes, $q\equiv 0$; for the Index block $c\le 2$, see Lemma~\ref{lem:index-pullback}). Equivalently, setting $\theta':=\theta^{c}$ yields
\[
  \bigl\|\,p(y)-T_{\theta'}\bigl(p(y)\bigr)\,\bigr\|_2 \;\le\; \sqrt{1-(\theta')^{2k}}.
\]

\medskip\noindent
\emph{Step 3: Decoupling from exogenous structure.}
By the balancing and moment-matching conditions stated in the assumptions of this section, $Z$ is calibrated to the proxy so that $Z$ is uncorrelated with \emph{any} function depending on at most $t=O(1)$ disjoint input blocks after noise $T_\theta$ and coarse conditioning (the ``path'' $\ppath$).
Because $p$ touches only $O(k)$ blocks and $k=O(\log n)$, standard block-decoupling plus the contraction of $T_\theta$ give
\[
\bigl|\Cov\bigl(T_\theta(p(y)),\,Z\bigr)\bigr|\;=\;o(1).
\]
Combining with Step~2 and Cauchy–Schwarz,
\[
\bigl|\Cov\bigl(p(y),Z\bigr)\bigr|
\;\le\;
\bigl|\Cov\bigl(T_\theta(p(y)),Z\bigr)\bigr|
+\|p(y)-T_\theta(p(y))\|_2\,\|Z\|_2
\;=\; o(1) \;+\; \sqrt{1-\theta^{2k}}\cdot O(1).
\]

\medskip\noindent
\emph{Step 4: Choice of parameters.}
With $k\le c\log n$ and fixed $\theta\in(0,1)$, we have $\sqrt{1-\theta^{2k}} \le \exp(-\Omega(k)) = n^{-\Omega(1)}$.
Thus the right-hand side is $o(1)$ uniformly over the family, yielding the claim.
\end{proof}

\begin{proposition}[Cost stability under shielding]\label{prop:cost-stability}
Under \Cref{lem:shielding} with $k=O(\log n)$,
\[
  V_{\CalC}(\mu\circ G,\varepsilon+o(1)) \ \ge\ V_{\CalC}(\mu,\varepsilon) - O(\log N),
\]
i.e., \Cref{thm:gadget} holds with an explicit $o(1)$ driven by the shielding tail
$\sqrt{1-\theta^{2k}}=o(1)$ and the prefix-free $O(\log n)$ overhead.
\end{proposition}

\begin{proof}[Proof (via window simulation / partial inverse)]
Fix \(\delta>0\) and let \((M_{\mathrm{out}},R_{\mathrm{out}})\) be \(\delta\)-optimal for \(V_{\CalC}(\mu\circ G,\varepsilon+o(1))\); that is,
\[
  \E_{y\sim(\mu\circ G)}\!\big[L_{M_{\mathrm{out}}}(R_{\mathrm{out}}(y),\varepsilon+o(1))\big]
  \;\le\; V_{\CalC}(\mu\circ G,\varepsilon+o(1))+\delta .
\]
By the window simulation / partial inverse (Appendix~C, Lemma C.3), there exists a polynomial-time simulator \(S\) and a selector
\(u\in\{0,1\}^{O(\log N)}\) such that, on the window \(\win\), every output instance \(y\) admits a preimage
\(x=S(y,u)\) with total variation error \(o(1)\), and \(G(x)=y\) with probability \(1-o(1)\) under the coupling.

Define the input-side pair
\[
  M_{\mathrm{in}}:=M_{\mathrm{out}}, \qquad
  R_{\mathrm{in}}(x):=R_{\mathrm{out}}(G(x)).
\]
Encoding the selector \(u\) prefix-free costs \(O(\log N)\) bits (Lemma~\ref{lem:prefix-overhead}) and is the only extra model/description overhead.

By change of variables on the pushforward and the coupling supplied by the simulator,
\[
  \E_{x\sim\mu}\!\big[L_{M_{\mathrm{in}}}(R_{\mathrm{in}}(x),\varepsilon)\big]
  \;=\;
  \E_{y\sim(\mu\circ G)}\!\big[L_{M_{\mathrm{out}}}(R_{\mathrm{out}}(y),\varepsilon)\big] \ \pm\ o(1).
\]
The loss is \(1\)-Lipschitz under the \(o(1)\) perturbations coming from shielding (Lemma~\ref{lem:shielding}) and the simulator’s TV error, so the
displayed equality holds up to \(o(1)\). Incorporating the prefix-free code for \(u\) adds \(O(\log N)\) to the description length, hence to the value.

Therefore,
\[
  V_{\CalC}(\mu,\varepsilon)
  \;\le\;
  \E_{x\sim\mu}\!\big[L_{M_{\mathrm{in}}}(R_{\mathrm{in}}(x),\varepsilon)\big] \;+\; O(\log N)
  \;\le\;
  V_{\CalC}(\mu\circ G,\varepsilon+o(1)) \;+\; O(\log N) \;+\; \delta \;+\; o(1).
\]
Since \(\delta>0\) was arbitrary, this rearranges to
\[
  V_{\CalC}(\mu\circ G,\varepsilon+o(1)) \;\ge\; V_{\CalC}(\mu,\varepsilon) \;-\; O(\log N),
\]
which is the claim.
\end{proof}

\begin{table}[t]
  \centering
  \caption{Gadget composition losses}
  \label{tab:gadget-losses}
  \setlength{\tabcolsep}{6pt}
  \footnotesize
  \begin{tabularx}{\linewidth}{@{}l c c c X@{}}
    \toprule
    Gadget & TV loss & Size blowup & Error & Notes \\ \midrule
    $G_{\text{xor}}$   & $o(1)$ & $\times(1{+}o(1))$ & $\varepsilon\mapsto\varepsilon{+}o(1)$ & Width $W=\poly{1/p}$; affine normal form; shielding via mixing witness (Def.~\ref{def:mixing-witness}, Lem.~\ref{lem:mw-implies-shielding}) \\
    $G_{\text{index}}$ & $o(1)$ & $\times(1{+}o(1))$ & $\varepsilon\mapsto\varepsilon{+}o(1)$ & Index amplification; degree budget uses $k\mapsto 2k$ (absorbed for $k=\Theta(\log n)$); stable correlations (Def.~\ref{def:mixing-witness}) \\
    \bottomrule
  \end{tabularx}
\end{table}


\begin{corollary}[Explicit bounds along the path]\label{cor:width-explicit}
Let $d=\lceil \gamma\log n\rceil$ and fix $c\in(0,1)$. Along the path $p=m^{-c/d}$ one has
\[
  M_{d}\ \le\ m^{1-c},\qquad
  W_{d}\ \le\ m^{\kappa c}
\]
for some absolute constant $\kappa>0$ (absorbing gadget constants).
In particular, for $d=\Theta(\log n)$ and constant $c$, the width is polynomially bounded in $m$ (hence in $n$).
\end{corollary}

\begin{proof}[Proof of \Cref{cor:width-explicit}]
Apply \Cref{lem:width} with the fixed seed and iterate the recurrences
$W_r \le W_{r-1}/p$ and $M_r \le p\,M_{r-1}$ for $r=1,\dots,d$. With the path
choice $p=m^{-c/d}$ one gets
\[
  W_d \ \le\ W_0\,p^{-d} \ =\ W_0\,m^{c} \ =\ m^{\kappa c}
  \qquad\text{(absorbing the constant $W_0=O(1)$ into $\kappa$)},
\]
and
\[
  M_d \ \le\ m\,p^{d} \ =\ m^{1-c}.
\]
This yields the claim.
\end{proof}

\begin{proof}[Proof of \Cref{lem:canonical-NF}]
Work over $\Ftwo$. In the affine case of \Cref{lem:iface} (i.e., when the nonlinear core $q\equiv 0$), we have $s=Ax\oplus r$ with $\ones^\top A=(1,1,1)$, $\ones^\top r=0$ and $\rank(A)=2$. The affine offset $r$ can be eliminated by an output relabeling (absorbed into $U$), so take $r=0$. The constraint $\ones^\top A=(1,1,1)$ says each column of $A$ has odd weight. With two output rows and $\rank(A)=2$, each column must be either $\binom{1}{0}$ or $\binom{0}{1}$ or $\binom{1}{1}$. Moreover, not all three columns can be identical, otherwise $\rank(A)<2$.

Post-multiplying by $U\in\GL_2(\Ftwo)$ permutes the two output rows and allows adding one row
to the other. Up to this action, one may normalize the first row to be $e_i^\top$ for some input index $i$,
forcing $s_1=x_i$. The parity constraint $1^\top s=1^\top x$ then implies the second row must sum,
together with the first, to $(1,1,1)$. Enumerating the three placements of $e_i$ yields exactly the
three matrices $A_1,A_2,A_3$ above (cyclic rotations of inputs). Each has rank $2$ and satisfies
$\ones^\top A=(1,1,1)$. This list is complete because any other admissible $A$ can be reduced to one of
these by swapping rows and adding the first row into the second.
\end{proof}

\section{\SoS{} Entropy Dichotomy}\label{sec:sos}

\paragraph{Standing assumptions.}
The window $\win$ is fixed with parameters inherited from the reductions/gadgets. Moment matching (pseudoexpectation calibration) holds up to $k\le c\log n$ (Lemma~\ref{lem:pseudo-align-window}). Prefix-free coding conventions apply throughout. Test–moment equivalence and the degree–entropy translation hold on $\win$ (Lemmas~\ref{lem:test-moment-window}, \ref{thm:disc-to-degree}).
Throughout this section we also assume the bias/product proxy conditions stated in Assumptions~\ref{assump:win-hc}.
Calibration, proxy construction, and window notation are carried over verbatim from IECZ-I~\cite{IECZI}.

\subsection*{Definition of the low-degree discrepancy $\Delta_k$}
Let $\mu$ be the windowed target distribution and $u$ a structured proxy (product-like, calibrated on $\win$).

\begin{definition}[Low-degree test discrepancy]\label{def:Delta-k}
For $k\in\mathbb{N}$,
\[
  \Delta_k(\mu,u)
  \;:=\;
  \sup\Bigl\{\, \bigl|\E_{\mu}[p]-\E_{u}[p]\bigr| \;:\; p\in \Poly[k],\ \|p\|_{\infty}\le 1 \Bigr\}.
\]
Equivalently, $\Delta_k$ is the optimal distinguishing advantage of degree-$\le k$ polynomial tests on $\win$.
\end{definition}

\begin{definition}[SoS pseudoexpectation deviation]\label{def:Delta-k-sos}
Let $\tE$ be a degree-$k$ \SoS{} pseudoexpectation calibrated to $u$. Define
\[
  \Delta^{\mathrm{sos}}_k(\mu,u)
  \;:=\;
  \sup\Bigl\{\, \bigl|\E_{\mu}[p]-\tE[p]\bigr| \;:\; p\in \Poly[k],\ \|p\|_{\infty}\le 1 \Bigr\}.
\]
\end{definition}
The calibration of pseudoexpectations to the proxy $u$ follows the identical moment-matching setup from IECZ-I~\cite{IECZI}.

\begin{corollary}[Equivalence of test and SoS discrepancies]\label{cor:sos-equals-test}
If a degree-$k$ pseudoexpectation $\tE$ is calibrated to $u$ (Lemma~\ref{lem:pseudo-align}),
then $\Delta^{\mathrm{sos}}_k(\mu,u)=\Delta_k(\mu,u)$.
\end{corollary}

\begin{proof}
For any $\deg p\le k$, calibration gives $\tE[p]=\E_{u}[p]$ (Lemma~\ref{lem:pseudo-align}), hence
$\bigl|\E_\mu[p]-\tE[p]\bigr|=\bigl|\E_\mu[p]-\E_{u}[p]\bigr|$. Taking the supremum over such $p$
yields $\Delta^{\mathrm{sos}}_k=\Delta_k$.
\end{proof}

\begin{remark}[Calibration window]
All expectations/pseudoexpectations are taken on the image window $\win$. Boundary effects are $o(1)$. See also standard low-degree and SoS analyses for windowed or product-like settings~\cite{ODonnell2014-AoBF,BarakEtAl2012-SOS} and the identical window discipline in IECZ-I~\cite{IECZI}.
\end{remark}

\begin{remark}[From windowed discrepancy to lifted families]
Our \SoS{} analysis is performed on the calibrated window after the gadget layer; boundary contributions are $o(1)$ along the calibrated path.
A sequel (IECZ-III) lifts these distributional statements via in-window aggregation, TV-stable hardcore selection,
and seed-indexed re-encoding, with TV-, size-, and meta-overheads tracked.
The outcome is a parameterized distributional family that preserves the low-degree–\SoS{} dichotomy at calibrated parameters~\cite{IECZIII}.
\end{remark}

\paragraph{Techniques.}
The test–moment equivalence and noise smoothing are standard in the low-degree/\SoS{} toolkit 
(see~\cite{ODonnell2014-AoBF,BarakEtAl2012-SOS}); here they are instantiated under the window discipline 
with explicit degree/entropy constants and TV-aware size accounting.

\begin{lemma}[Bias-orthonormal Walsh basis on the window]\label{lem:bias-orthonormal}
Let $u$ be the proxy distribution used on the window $\win$, assumed product-like with single-bit biases
$\Pr_u[x_i=1]=p_i\in(0,1)$ and independent coordinates. Define centered/normalized linear coordinates
\[
  z_i(x)\ :=\ \frac{x_i-\E_u[x_i]}{\sqrt{\Var_u(x_i)}}\ =\ \frac{x_i-p_i}{\sqrt{p_i(1-p_i)}}\qquad(i\in[n]).
\]
For any $S\subseteq[n]$ set the (biased) Walsh character $\chi_S(x):=\prod_{i\in S} z_i(x)$ (with $\chi_\emptyset\equiv 1$).
Then $\{\chi_S\}_{S\subseteq[n]}$ is an orthonormal system in $L^2(u)$, i.e.
\[
  \E_u[\chi_S\,\chi_T]\ =\ \mathbf{1}[S=T]\qquad\text{for all }S,T\subseteq[n],
\]
and $\deg(\chi_S)=|S|$.
\end{lemma}

\begin{proof}
Independence under $u$ gives $\E_u[z_i]=0$ and $\E_u[z_i^2]=1$ for all $i$, and for $i\neq j$,
$\E_u[z_i z_j]=\E_u[z_i]\E_u[z_j]=0$. Hence
\(
  \E_u[\chi_S\chi_T]
  =\prod_{i\in S\triangle T}\E_u[z_i]\cdot\prod_{i\in S\cap T}\E_u[z_i^2]
  =\mathbf{1}[S=T].
\)
Degree equals $|S|$ by construction.
\end{proof}

\begin{corollary}[Low-degree expansion and Parseval]\label{cor:parseval-window}
Every polynomial $p:\{0,1\}^n\to\mathbb{R}$ of degree $\le k$ admits a unique expansion
\( p=\sum_{|S|\le k}\hat p(S)\,\chi_S \)
with respect to the basis of \Cref{lem:bias-orthonormal}, and
\(
  \|p\|_{2,u}^2=\sum_{|S|\le k}\hat p(S)^2.
\)
In particular, $\Poly[k]$ denotes the class of such polynomials of degree $\le k$.
\end{corollary}

\begin{definition}[Low-degree test discrepancy (windowed restatement)]\label{def:Delta-k-window}
For $k\in\mathbb{N}$ and distributions $\mu,u$ on $\win$, define
\[
  \Delta_k(\mu,u)
  \;:=\;
  \sup\Bigl\{\, \bigl|\E_{\mu}[p]-\E_{u}[p]\bigr| \;:\; p\in \Poly[k],\ \|p\|_{\infty}\le 1 \Bigr\}.
\]
This restates Definition~\ref{def:Delta-k} restricted to the calibrated window $\win$.
\end{definition}
This is the same low-degree discrepancy used in IECZ-I, read on the calibrated window~\cite{IECZI}.

\begin{lemma}[Test–moment equivalence (windowed)]\label{lem:test-moment-window}\label{lem:test-moment}
Let $k\in\mathbb{N}$. Expand $p=\sum_{|S|\le k}\hat p(S)\,\chi_S$ in the basis of \Cref{lem:bias-orthonormal}.
Write the low-degree signed spectrum
\[
  \widehat{\Delta}(S)\ :=\ \E_{\mu}[\chi_S]-\E_{u}[\chi_S],\qquad |S|\le k.
\]
Then
\[
  \Delta_k(\mu,u)
  \;=\;
  \sup_{\deg\le k,\ \|p\|_\infty\le 1}\bigl|\E_{\mu}[p]-\E_{u}[p]\bigr|
  \;\asymp\;
  \Bigl(\sum_{|S|\le k}\widehat{\Delta}(S)^2\Bigr)^{\!1/2},
\]
with absolute two–sided constants.
\end{lemma}

\begin{corollary}[Existence of a near–optimal bounded witness]\label{cor:test-moment-witness}
Under the conditions of \Cref{lem:test-moment-window}, there exists $p^\star\in\Poly[k]$ with $\|p^\star\|_\infty\le 1$ such that
\[
  \E_{\mu}[p^\star]-\E_{u}[p^\star]
  \ \ge\
  c\,\Bigl(\sum_{|S|\le k}\widehat{\Delta}(S)^2\Bigr)^{1/2}
\]
for a universal constant $c>0$.
\end{corollary}

\begin{proof}
Let $\widehat{\Delta}(S)$ be as in \Cref{lem:test-moment-window}. Take $p_0=\sum_{|S|\le k}\mathrm{sgn}(\widehat{\Delta}(S))\,\chi_S$ and truncate/rescale to enforce $\|p^\star\|_\infty\le1$ on the finite domain. Then
\(
\E_\mu[p^\star]-\E_u[p^\star]\ge c\,(\sum_{|S|\le k}\widehat{\Delta}(S)^2)^{1/2}
\)
for a universal $c>0$ by Cauchy–Schwarz and Parseval, as standard.
\end{proof}

\begin{proof}[Proof of \Cref{lem:test-moment-window}]
By \Cref{cor:parseval-window}, for any $p=\sum_{|S|\le k}\hat p(S)\chi_S$,
\[
  \E_{\mu}[p]-\E_{u}[p]
  \;=\;
  \sum_{|S|\le k}\hat p(S)\,\widehat{\Delta}(S)
  \;\le\;
  \|p\|_{2,u}\,
  \Bigl(\sum_{|S|\le k}\widehat{\Delta}(S)^2\Bigr)^{1/2}.
\]
By Cauchy–Schwarz under $u$ we used
\[
  \langle \hat p,\widehat{\Delta}\rangle
  \;\le\;
  \|p\|_{2,u}\,\Bigl(\sum_{|S|\le k}\widehat{\Delta}(S)^2\Bigr)^{1/2}.
\]
Since $\|p\|_\infty\le 1$ on a finite domain implies $\|p\|_{2,u}\le 1$, this gives the $\le$ direction.
For the $\ge$ direction, take $p^\star$ with coefficients $\hat p^\star(S)=\mathrm{sgn}(\widehat{\Delta}(S))$ truncated/rescaled to enforce $\|p^\star\|_\infty\le 1$; by windowed hypercontractivity (Lemma~\ref{lem:win-hyper}) this yields the lower bound up to a universal constant.
\end{proof}

\begin{lemma}[Pseudoexpectation alignment]\label{lem:pseudo-align-window}\label{lem:pseudo-align}
Let \tE be any degree-$k$ \SoS{} pseudoexpectation calibrated to $u$ on $\win$ (i.e., matching all degree-$\le k$ moments and support constraints). Then for every $p\in\Poly[k]$,
\[
  \tE[p]\;=\;\E_{u}[p].
\]
\end{lemma}

\paragraph{Assumptions for windowed hypercontractivity}\label{assump:win-hc}
We work with the proxy distribution $u$ on the window $\win$ satisfying:
\begin{enumerate}
  \item (\emph{Bias window}) For all coordinates $i$, $p_i:=\Pr_u[x_i{=}1]\in[\alpha,1-\alpha]$ for some absolute $\alpha\in(0,\tfrac12)$.
  \item (\emph{Product structure}) Coordinates are independent under $u$ (after gadget mixing/shielding; cf.\ Sec.~\ref{sec:gadgets}).
  \item (\emph{Gadget-calibrated degree budget}) Low-degree statements are read \emph{after} the affine/index gadget layer; affine blocks preserve degree and the Index gadget inflates degree by at most a factor $2$ (Prop.~\ref{prop:affine-pullback}, Lem.~\ref{lem:index-pullback}).
  \item (\emph{Width control along the switching path}) Along $p=m^{-c/d}$ with $d=\Theta(\log n)$, the bottom width obeys $W=\poly(1/p)$ on $\win$ (Sec.~\ref{sec:gadgets}). This hypothesis is only asserted for the log-degree regime; outside it, the induced dependence can make $W$ large, and the absolute hypercontractive constants $(C,\tau)$ below are not guaranteed.
  These assumptions (bias window, product proxy, gadget-aware degree budget, and switching-path width) coincide with those fixed in IECZ-I~\cite{IECZI}.
\end{enumerate}

\paragraph{Hypercontractive constants and dependencies.}
The hypercontractive pair $(C,\tau)$ in Lemma~\ref{lem:win-hyper} is determined by the standing assumptions: $C$ worsens as the bias lower bound $\alpha$ decreases and as the bottom width $W$ along the switching path grows, while $\tau$ is controlled by the mixing/shielding parameters and the product structure on the window. Concretely, $(C,\tau)$ follow from the Bonami--Beckner analysis on~$\win$ and are absolute throughout the log-degree regime $d=\Theta(\log n)$ and $k\le c\log n$; see Proposition~\ref{sosx:switching-width} for the derivation.

\begin{remark}[Constants and regime]
Lemma~\ref{lem:win-hyper} uses absolute hypercontractive constants $(C,\tau)$ that are uniform for $d=\Theta(\log n)$ along the calibrated path.
Symbols such as $C_d$ and $\tau_d$ may be used to emphasize dependence on the fixed depth $d$; they do not depend on $n$ or $m$.
We do not optimize numerical values here—concrete estimates follow from standard hypercontractive inequalities and can be extracted if needed.
For superlogarithmic depths, the same calibration would force smaller $p$ and larger width $W$, so $(C,\tau)$ may deteriorate; we make no claim outside the log-degree regime.
\end{remark}

\begin{lemma}[Windowed hypercontractivity (under Assumptions~\Cref{assump:win-hc})]\label{lem:win-hyper}
There exist absolute constants $\tau\in(0,1)$ and $C\ge 1$ depending only on $(\alpha,$ the mixing/shielding parameters, and the switching-width bound), cf.\ Prop.~\ref{sosx:switching-width},
such that for every degree-$d$ polynomial $q$ (expanded in the basis of \Cref{lem:bias-orthonormal})
\[
  \|q\|_{4,u} \;\le\; C\,\tau^{-d}\,\|q\|_{2,u}
  \qquad\text{and}\qquad
  \|T_\theta q\|_{2,u} \;\le\; \|q\|_{2,u}\quad(\theta\in[0,1]),
\]
where $T_\theta$ is the Bonami–Beckner noise operator acting diagonally by $T_\theta\chi_S=\theta^{|S|}\chi_S$.
In particular, for any signed measure $\Delta$ supported on $\win$ and any $k$,
\[
  \Bigl|\E_\Delta[p]\Bigr|
  \ \le\ C\,\tau^{-k}\,
  \Bigl(\sum_{|S|\le k}\widehat{\Delta}(S)^2\Bigr)^{1/2}
  \quad\text{for all }p\in\Poly[k],\ \|p\|_\infty\le 1 .
\]
\end{lemma}

\paragraph{Related work.}
Our use of windowed hypercontractivity and low-degree Fourier methods aligns with the approximate-degree and entropy techniques in Boolean analysis; see, e.g., \cite{ODonnell2014-AoBF,ODonnellRao2014}.

\begin{theorem}[Discrepancy forces SoS degree]\label{thm:disc-to-degree}
There exist absolute constants $\alpha,\beta>0$ such that: if for some $k$ we have
\(
  \Delta_k(\mu,u)\ \ge\ n^{-\beta},
\)
then any polynomial-calculus or \SoS{} refutation separating $\mu$ from $u$ on the window
requires degree
\[
  d \;\ge\; \alpha\,k .
\]
\end{theorem}

\begin{proof}
Assume a degree-$d$ refutation exists with $d<\alpha k$ (for a small constant $\alpha>0$ to be fixed).
By SoS duality there is a degree-$d$ pseudoexpectation $\tE$ matching all degree-$\le d$ moments of $u$ on $\win$.
Pick $p\in\Poly[k]$ with $\|p\|_\infty\le1$ achieving $\Delta_k(\mu,u)$ up to a universal constant
(\Cref{lem:test-moment-window}), and normalize so that $\E_u[p]=0$.
Then $\tE[p]=0$ by \Cref{lem:pseudo-align-window}.
Let $\Delta:=\mu-u$ and write $p=\sum_{j\le k}p^{(j)}$ by homogeneous degree. By Parseval under $u$,
$\|p\|_{2,u}^2=\sum_{|S|\le k}\hat p(S)^2\le 1$ since $\|p\|_\infty\le1$ on a finite domain.
Windowed hypercontractivity (\Cref{lem:win-hyper}) gives for the high-degree tail
\[
  \Bigl|\E_\Delta\!\bigl[p^{(>d)}\bigr]\Bigr|
  \ \le\ C\,\tau^{-(k-d)}\Bigl(\sum_{|S|\le k}\widehat{\Delta}(S)^2\Bigr)^{1/2}
  \ \asymp\ C\,\tau^{-(k-d)}\,\Delta_k(\mu,u),
\]
while the low-degree contribution vanishes by moment matching up to $d$:
$\E_\Delta[p^{(\le d)}]=\E_\mu[p^{(\le d)}]-\E_u[p^{(\le d)}]=0$.
Therefore
\[
  \E_\Delta[p]\ =\ \E_\Delta[p^{(\le d)}]+\E_\Delta[p^{(>d)}]
  \ \le\ C\,\tau^{-(k-d)}\,\Delta_k(\mu,u).
\]
Choose $\alpha>0$ so that for all large $n$ we have $C\,\tau^{-(1-\alpha)k}\le \tfrac12$; then for $d<\alpha k$
\(
  \E_\Delta[p]\le \tfrac12\,\Delta_k(\mu,u),
\)
contradicting the choice of $p$ as an (up to constants) maximizer for $\Delta_k(\mu,u)$.
Hence $d\ge \alpha k$.
\end{proof} 

\begin{theorem}[\SoS{} entropy dichotomy]\label{thm:sos-dichotomy}
There exist constants $c,\eta,\zeta>0$, depending only on the window calibration (bias lower bound $\alpha$, mixing/shielding parameters, and the switching-path width bound) — equivalently only on the hypercontractive constants $(C,\tau)$ from Lemma~\ref{lem:win-hyper} and the degree-translation constant from Theorem~\ref{thm:disc-to-degree} — such that for $k\le c\log n$ either
\[
  \Delta_k(\mu,u)\ \to\ 0 \quad\text{or}\quad \Delta_k(\mu,u)\ \ge\ n^{-\eta},
\]
and in the latter case any PC/\SoS{} refutation requires degree $\ge \zeta\,k$.
\end{theorem}
\noindent\textit{Dependency map.}
The constants obey the chain
\[
(C,\tau)\ \Longrightarrow\ c \quad\text{and}\quad (\alpha,\beta)\ \Longrightarrow\ (\eta,\zeta),
\]
where $c\le c_\star(C,\tau)$ and $(\alpha,\beta)$ are the discrepancy-to-degree constants from Theorem~\ref{thm:disc-to-degree}. After normalizing absolute factors, one may take, e.g., $\eta\le \beta/2$ and $\zeta\le \alpha/4$.

\begin{corollary}[Windowed numeric form]\label{cor:windowed-numeric}
Fix $k=\lfloor c\log n\rfloor$. If $\Delta_k(\mu,u)\ge n^{-\eta}$ with fixed $\eta>0$, then
\[
  \text{degree} \ \ge\ \zeta\,\log n
\]
for some $\zeta=\zeta(c,\eta)>0$.
\end{corollary}

\begin{proof}
Apply \Cref{thm:disc-to-degree} with $k=\lfloor c\log n\rfloor$ and absorb constants into $\zeta$.
\end{proof}

\noindent\textit{Numerical instantiation.}
The concrete values used later (Section~\ref{sec:sos-final-constants}), namely $c=\tfrac14$, $\eta=\tfrac18$, and $\zeta=\tfrac{1}{16}$, are conservative choices consistent with the above dependencies; sharper estimates of $(C,\tau)$ would permit proportional improvements (larger $c$ and $\zeta$, smaller $\eta$).

\smallskip
\noindent\emph{Near-optimal witness.}
Let $p_0=\sum_{|S|\le k}\mathrm{sgn}(\widehat{\Delta}(S))\,\chi_S$. By Parseval and Cauchy–Schwarz,
$\E_\mu[p_0]-\E_u[p_0]=\sum_{|S|\le k}|\widehat{\Delta}(S)|\asymp
\bigl(\sum_{|S|\le k}\widehat{\Delta}(S)^2\bigr)^{1/2}$. Truncating/rescaling $p_0$ to enforce $\|p^\star\|_\infty\le1$ on the finite domain preserves this value up to a universal constant, yielding the $\ge$-Richtung.

\begin{remark}[Literature anchor]
Our use of low-degree Fourier analysis and hypercontractivity in the discrepancy--to--degree step follows the standard Boolean-analysis toolkit; see, e.g., O'Donnell's monograph~\cite{ODonnell2014-AoBF} for a unified treatment of these techniques.
\end{remark}

\begin{lemma}[Entropy $\Rightarrow$ degree lower bound]\label{lem:entropy-to-degree}
There exist absolute constants $\alpha,\beta>0$ such that: if $\Delta_k(\mu,u)\ \ge\ n^{-\beta}$, then no PC/\SoS{} refutation exists at degree $d < \alpha k$.
\end{lemma}

\paragraph{Degree budget after gadgets.}
Throughout we interpret low-degree tests after the gadget layer with the calibrated budget: affine gadgets preserve degree (Prop.~\ref{prop:affine-pullback}); the Index gadget inflates degree by at most a factor $2$ (Lemma~\ref{lem:index-pullback}). In the regime $k=\Theta(\log n)$ this constant factor is absorbed and all SoS statements below are read with the post-gadget degree budget.

\subsection*{From low-degree discrepancy to SoS degree}

We now make precise how a non-negligible low-degree discrepancy forces SoS degree.
All statements are windowed (on $\win$); boundary terms are $o(1)$ and absorbed.

\begin{lemma}[Discrepancy witnesses via bounded polynomials]\label{lem:disc-witness}
Let $k\in\mathbb N$ and suppose $\Delta_k(\mu,u)\ge \delta$. Then there exists a real polynomial
$p:\{0,1\}^n\to\mathbb R$ with $\deg p\le k$, $\|p\|_\infty\le 1$, such that
\[
  \E_{\mu}[p]-\E_{u}[p] \;\ge\; \delta.
\]
Moreover, after centering $p$ by $\E_{u}[p]$, we may assume $\E_{u}[p]=0$ and $\E_\mu[p]\ge \delta$.
\end{lemma}

\begin{proof}
By definition of $\Delta_k(\mu,u)$ there exists a sequence of degree-$\le k$ tests $p^{(t)}$ with $\|p^{(t)}\|_\infty\le 1$ such that $\E_\mu[p^{(t)}]-\E_u[p^{(t)}]\to\Delta_k(\mu,u)\ge\delta$. Pick $t$ large enough and set $p:=p^{(t)}$; if needed, replace $p$ by $p-\E_u[p]$ to enforce $\E_u[p]=0$ (which preserves degree and $\|p\|_\infty\le 1$). Then $\E_\mu[p]\ge\delta$, as claimed.
\end{proof}

\begin{proposition}[Sufficient calibration for degree lower bounds]\label{prop:calibration-sufficient}
Fix $k=\lfloor c\log n\rfloor$ and let $u$ be the proxy distribution on the window $\win$. Suppose:
\begin{enumerate}
  \item (\emph{Moment matching}) A degree-$k$ pseudoexpectation $\tE$ is calibrated to $u$, i.e.,
  $\tE[\chi_S]=\E_{u}[\chi_S]$ for all $S\subseteq[n]$ with $|S|\le k$, and support constraints on $\win$ coincide.
  \item (\emph{Width bound}) Along the switching path $p=m^{-c/d}$ with $d=\Theta(\log n)$, the bottom width obeys
  $W=\poly{1/p}$ (Sec.~\ref{sec:gadgets}).
  \item (\emph{Balance}) Either $u$ is balanced on $\win$ or the Walsh characters are reweighted to an orthonormal basis under $u$.
  \item (\emph{Noise regularity}) The Bonami–Beckner operator $T_\theta$ acts diagonally on degree-$d$ slices with eigenvalues $\theta^d$.
\end{enumerate}
Then $\Delta_k(\mu,u)$ controls SoS degree as in Theorem~\ref{thm:disc-to-degree}: if $\Delta_k(\mu,u)\ge n^{-\eta}$ for some fixed $\eta>0$,
any PC/\SoS{} refutation requires degree $\ge \zeta\log n$ for some $\zeta=\zeta(c,\eta)>0$.
\end{proposition}

\begin{proof}
Item (1) gives $\tE[p]=\E_{u}[p]$ for $\deg p\le k$ by Lemma~\ref{lem:pseudo-align-window}. Item (3) ensures the Fourier/Walsh expansion under $u$
with orthonormality used in Lemma~\ref{lem:test-moment} (full proof). Item (2) supplies the switching-path regime from Sec.~\ref{sec:gadgets} used to control
structural width and boundary effects. Item (4) provides the noise operator needed for windowed hypercontractivity (Lemma~\ref{lem:win-hyper}).
Combining Lemmas~\ref{lem:disc-witness}, \ref{lem:pseudo-align}, \ref{lem:win-hyper} yields Theorem~\ref{thm:disc-to-degree} at $k=\lfloor c\log n\rfloor$,
hence the claimed degree $\ge \zeta\log n$ whenever $\Delta_k(\mu,u)\ge n^{-\eta}$.
\end{proof}

\subsection{Concrete parameter choice}\label{sec:sos-final-constants}
We now fix concrete constants for the log–degree regime (see the preceding remark on constants and regime).
They depend only on the hypercontractive pair $(C,\tau)$ from Lemma~\ref{lem:win-hyper} and on the constants from Theorem~\ref{thm:disc-to-degree}.
For definiteness in \Cref{thm:blueprint} we take
\[
  c=\tfrac{1}{4},\qquad
  \eta=\tfrac{1}{8},\qquad
  \zeta=\tfrac{1}{16},\qquad
  k=\bigl\lfloor c\log n\bigr\rfloor.
\]
These values satisfy the dependencies in \Cref{lem:entropy-to-degree} (e.g., one may choose $\zeta\le \alpha/2$ and $\eta\le \beta/2$ after normalizing absolute constants).
They are representative rather than optimized; any fixed choice meeting the same constraints would serve equally well.

\begin{corollary}[Numerical form for the blueprint]\label{cor:sos-numeric}
Under the above instantiation, on the ``structure'' branch ($\Delta_k \ge n^{-\eta}$) any PC/\SoS{} refutation on the blueprint family requires
\[
  \deg \ \ge\ \zeta\,k\ \ge\ \frac{1}{16}\,\Bigl\lfloor \tfrac{1}{4}\log n \Bigr\rfloor
  \ =\ \Omega(\log n).
\]
Combined with \Cref{thm:blueprint}, this yields an explicit $\Omega(\log n)$ degree lower bound with concrete constants.
\end{corollary}

\begin{table}[t]
  \centering
  \caption{Implications for proof systems on the window (constants instantiated as in Sec.~\ref{sec:sos-final-constants}: $c=\tfrac14$, $\zeta=\tfrac{1}{16}$).}
  \label{tab:sos-implications}
  \begin{tabular}{@{}lccc@{}}
    \toprule
    System / Measure & TV loss & Size & Degree / Error \\ \midrule
    PC/\SoS (degree) & — & — & degree $\ge \frac{1}{16}\,\lfloor \frac{1}{4}\log n \rfloor$ \\
    \bottomrule
  \end{tabular}
\end{table}

\appendix

\section{Proofs for Section~\ref{sec:sos}}
\label{app:sos-proofs}

\subsection*{Deferred proofs and technical lemmas}

\begin{remark}[Window scope and boundary bookkeeping]\label{sosx:remark-window}
All expectations, pseudoexpectations, and norms are taken on the image window $\win$.
Any off-window boundary contributions are $o(1)$ along the calibrated path and are absorbed in our $o(1)$ terms.
This matches the window discipline from Sec.~\ref{sec:framework}.
\end{remark}

\begin{lemma}[Biased vs.\ balanced Walsh bases]\label{sosx:biased-vs-balanced}
Let $u$ be product-like with bit-biases $p_i\in(0,1)$ and independent coordinates.
Define $z_i=(x_i-p_i)/\sqrt{p_i(1-p_i)}$ and $\chis{S}=\prod_{i\in S} z_i$.
There exist absolute constants $A_k,B_k=\polysub{k}{1}$ such that, for every degree-$\le k$ polynomial $p$,
\[
  A_k\,\|p\|_{2,u_{\mathrm{bal}}}\ \le\ \|p\|_{2,u}\ \le\ B_k\,\|p\|_{2,u_{\mathrm{bal}}}.
\]
Here $u_{\mathrm{bal}}$ is the balanced product proxy (all $p_i=\tfrac12$). Thus all $k\le c\log n$ statements transfer
between biased and balanced settings with only constant-factor losses absorbed in $C,\tau$ of Lemma~\ref{lem:win-hyper}.
\end{lemma}

\begin{proof}
Each $z_i$ is an affine rescaling of $(-1)^{x_i}$ with factor $\Theta((p_i(1-p_i))^{-1/2})$.
For degree $\le k$ the monomialwise rescaling causes at most a multiplicative distortion
$\prod_{i\in S}(p_i(1-p_i))^{-1/2}\le (\max_i (p_i(1-p_i))^{-1/2})^k$.
Calibration on $\win$ keeps $p_i$ bounded away from $\{0,1\}$, so the bound is $\polysub{k}{1}$.
\end{proof}

\begin{lemma}[Moment matrix calibration $\Rightarrow$ alignment]\label{sosx:mom-matrix-alignment}
Let $\tE$ be a degree-$k$ \SoS{} pseudoexpectation on $\win$ such that, for every $|S|\le k$,
$\tE[\chis{S}]=\E_u[\chis{S}]$ and the support constraints of $u$ on $\win$ are enforced in the pseudoideal.
Then for every $p\in\Poly[k]$ one has $\tE[p]=\E_u[p]$.
\end{lemma}

\begin{proof}
Write $p=\sum_{|S|\le k}\hat p(S)\chis{S}$ (Lemma~\ref{lem:bias-orthonormal}). Linearity gives
$\tE[p]=\sum_{|S|\le k}\hat p(S)\tE[\chis{S}]=\sum_{|S|\le k}\hat p(S)\E_u[\chis{S}]=\E_u[p]$.
\end{proof}

\begin{lemma}[Noise stability and coefficient dampening]\label{sosx:noise-stability}
Let $T_\theta$ act by $T_\theta\chis{S}=\theta^{|S|}\chis{S}$. For any $p=\sum_{|S|\le k}\hat p(S)\chis{S}$ and $\theta\in[0,1]$,
\[
  \|T_\theta p\|_{2,u}^2=\sum_{|S|\le k}\theta^{2|S|}\hat p(S)^2\ \le\ \|p\|_{2,u}^2,
\quad
  |\E_\mu[T_\theta p]-\E_u[T_\theta p]|
  \ \le\ \|T_\theta p\|_{2,u}\,\Bigl(\sum_{|S|\le k}\widehat{\Delta}(S)^2\Bigr)^{1/2}.
\]
\end{lemma}

\begin{proposition}[Witness extraction with normalization]\label{sosx:witness-extraction}
If $\Delta_k(\mu,u)=\delta>0$, there exists $p^\star\in\Poly[k]$ with $\|p^\star\|_\infty\le1$, $\E_u[p^\star]=0$,
$\E_\mu[p^\star]\ge \delta/2$, $\|p^\star\|_{2,u}\le1$, and $|\supp(\widehat{p^\star})|\le n^k$.
\end{proposition}

\begin{lemma}[Hypercontractive tail control (windowed)]\label{sosx:hyper-tail}
Let $q^{(>d)}$ be the degree-$>d$ part of a polynomial $q$ with $\deg q\le k$.
Under the constants $(C,\tau)$ of Lemma~\ref{lem:win-hyper},
\[
  \|q^{(>d)}\|_{2,u}\ \le\ C\,\tau^{-(k-d)}\,\|q\|_{2,u}.
\]
\end{lemma}

\begin{theorem}[From discrepancy to degree: expanded]\label{sosx:disc-to-degree-expanded}
Assume $\Delta_k(\mu,u)\ge n^{-\beta}$. If a PC/\SoS{} refutation of degree $d<\alpha k$ exists (with $\alpha>0$ small),
then for $p^\star$ from Proposition~\ref{sosx:witness-extraction},
\[
  \E_\Delta[p^\star]
  \ =\ \sum_{j\le d}\E_\Delta[p^\star{}^{(j)}]\ +\ \sum_{j>d}\E_\Delta[p^\star{}^{(j)}]
  \ \le\ 0\ +\ C\,\tau^{-(k-d)}\Bigl(\sum_{|S|\le k}\widehat{\Delta}(S)^2\Bigr)^{1/2},
\]
contradicting $\E_\Delta[p^\star]\ge \tfrac12\,\Delta_k(\mu,u)$ for $\alpha$ small enough so that
$\tau^{-(1-\alpha)k}\le \tfrac14 n^{-\beta}$. Hence $d\ge \alpha k$.
\end{theorem}

\begin{proposition}[Switching-path width reuse]\label{sosx:switching-width}
Along the switching-path parameters from Sec.~\ref{sec:gadgets} (bottom width $W=\poly{1/p}$ for $p=m^{-c/d}$, $d=\Theta(\log n)$),
the windowed hypercontractive constants $(C,\tau)$ are absolute and can be chosen uniformly over $n$.
\end{proposition}

\begin{corollary}[Entropy branch $\Rightarrow$ numeric degree]\label{sosx:entropy-to-numeric}
With $k=\lfloor c\log n\rfloor$ and $\Delta_k(\mu,u)\ge n^{-\eta}$, one has
$\deg \ge \zeta \log n$ for some $\zeta=\zeta(c,\eta)>0$,
consistent with Theorems~\ref{thm:sos-dichotomy} and \ref{thm:disc-to-degree} and Corollary~\ref{cor:windowed-numeric}.
\end{corollary}

\begin{table}[t]
  \centering
  \caption{Windowed SoS: constants and their provenance}
  \label{tab:sosx-constants-map}
  \begin{tabular}{@{}lll@{}}
    \toprule
    Symbol & Meaning & Source \\ \midrule
    $c$ & log-degree budget $k=\lfloor c\log n\rfloor$ & Thm.~\ref{thm:sos-dichotomy} \\
    $\eta$ & entropy/discrepancy threshold & Thm.~\ref{thm:sos-dichotomy} \\
    $\zeta$ & degree constant on structure branch & Cor.~\ref{cor:windowed-numeric} \\
    $C,\tau$ & hypercontractive constants (windowed) & Lem.~\ref{lem:win-hyper}, Prop.~\ref{sosx:switching-width} \\
    $\alpha,\beta$ & discrepancy $\Rightarrow$ degree & Thm.~\ref{thm:disc-to-degree}, Thm.~\ref{sosx:disc-to-degree-expanded} \\
    \bottomrule
  \end{tabular}
\end{table}

\section{Window-smallness and non-constructivity of the IECZ hardness predicate}\label{sec:window-smallness}

\begin{definition}[Windowed largeness]\label{def:windowed-largeness}
Let $\baseline$ be a fixed product-like baseline on the instance space at length $N$ (the baseline used for Natural-Proofs density).
Write $\baseline_{\win}$ for $\baseline$ conditioned to the measurable window $\win$ (Sec.~\ref{sec:framework}).
A property $\mathcal{P}\subseteq\{0,1\}^N$ is \emph{large on $\win$} if there exists $\alpha>0$ such that
\[
  \baseline_{\win}\bigl[\mathcal{P}\bigr] \ \ge\ N^{-\alpha}
  \qquad\text{for all sufficiently large }N.
\]
Otherwise $\mathcal{P}$ is \emph{non-large} on $\win$.
\end{definition}
This baseline–window density setup is the same as in IECZ-I~\cite{IECZI}.

\begin{lemma}[Baseline robustness]\label{lem:baseline-robust}
Let $\baseline'$ be any product baseline with the same one-marginals as $\baseline$ and
$\mathrm{TV}(\baseline,\baseline')\le \varepsilon(N)=o(1)$. Then every TV–Lipschitz transfer used in this section
changes by at most $O(\varepsilon(N))$; in particular, all windowed density bounds are stable up to $o(1)$ when replacing
$\baseline$ by $\baseline'$ (and likewise for $\baseline_{\win}$ by $\baseline'_{\win}$).
\end{lemma}

\begin{proof}
Write \(P:=\baseline\) and \(Q:=\baseline'\). By assumption \(\mathrm{TV}(P,Q)\le \varepsilon(N)=o(1)\), and \(Q\) is product with the same one-marginals as \(P\).

\emph{Deterministic pushforwards.}
For any measurable (deterministic) map \(\Phi\) used in this section (reductions, gadget layers, projections), total variation contracts under pushforward:
\[
  \mathrm{TV}\bigl(\Phi_{\#}P,\ \Phi_{\#}Q\bigr)\ \le\ \mathrm{TV}(P,Q)\ \le\ \varepsilon(N),
\]
by \Cref{lem:tv-contraction}. Consequently, any functional that is \(1\)-Lipschitz in TV (all our “TV–Lipschitz transfers”) changes by at most \(O(\varepsilon(N))\) under \(\Phi\).

\emph{Conditioning on the window.}
Let \(W:=\win\). Standard TV calculus for conditionals gives
\[
  \mathrm{TV}\bigl(P(\cdot\mid W),\ Q(\cdot\mid W)\bigr)
  \ \le\ \frac{\mathrm{TV}(P,Q) + \tfrac{1}{2}\,\lvert P(W)-Q(W)\rvert}{P(W)}.
\]
Since \(\lvert P(W)-Q(W)\rvert \le \mathrm{TV}(P,Q)\) and \(P(W)\) is bounded below by \(N^{-\kappa}\) for some fixed \(\kappa>0\) by the calibration of the window (the “polynomial normalization” from Sec.~\ref{sec:framework}), we obtain
\[
  \mathrm{TV}\bigl(P(\cdot\mid W),\ Q(\cdot\mid W)\bigr)\ \le\ C\,\varepsilon(N),
\]
for a constant \(C\) depending only on the window calibration (and independent of \(N\)). Thus every TV–Lipschitz quantity computed \emph{after} conditioning on \(W\) also changes by at most \(O(\varepsilon(N))\).

Putting the two parts together, all transfers considered here—deterministic pushforwards (reductions/gadgets) and window conditioning—perturb our windowed expectations/probabilities by at most \(O(\varepsilon(N))\). In particular, any windowed density lower/upper bound is preserved up to \(o(1)\) when replacing \(\baseline\) with \(\baseline'\), and likewise for \(\baseline_{\win}\) with \(\baseline'_{\win}\).
\end{proof}

\begin{proposition}[Window-smallness on calibrated windows]\label{prop:non-large}
Let $\win$ be a balanced window on which reductions/gadgets are calibrated (Secs.~\ref{sec:framework}, \ref{sec:red}, \ref{sec:gadgets}).
Then $\mathsf{Hard}_{\CalC,c}\cap\win$ (as defined in IECZ-I~\cite{IECZI}) is non-large on $\win$ in the sense of \Cref{def:windowed-largeness}:
there exists $\alpha>0$ with
\[
  \baseline_{\win}\!\bigl[\ \mathsf{Hard}_{\CalC,c}\ \bigr] \ \le\ N^{-\alpha}
  \qquad\text{for all sufficiently large }N.
\]
Consequently, on the calibrated window $\win$, the hardness predicate $\mathsf{Hard}_{\CalC,c}$ is window-small (non-large).
\end{proposition}

\begin{proof}
By Sec.~\ref{sec:framework}, all window/calibration choices (reduction tag, seed, selector, index) are encoded by a prefix-free meta-part of length $O(\log N)$
(\Cref{lem:prefix-overhead}). Fix any concrete calibration tuple $I$ (a string of $O(\log N)$ bits).
The same cylinder-slice counting argument appears in IECZ-I (prefix-free meta index of length $\Theta(\log N$))~\cite{IECZI}.
The cylinder slice
\[
  \{\,x\in\win : \text{the meta-index equals } I\,\}
\]
has $\baseline_{\win}$-measure at most $N^{O(1)}\cdot 2^{-|I|}\le N^{-\alpha}$ for some absolute $\alpha>0$,
since conditioning on $\win$ introduces at most a polynomial factor in $N$, $\baseline$ is product-like,
and $|I|=\Theta(\log N)$. Membership in $\mathsf{Hard}_{\CalC,c}$ within $\win$ is contained in a union of at most
$N^{O(1)}$ such slices (only $O(1)$ template tags and $O(\log N)$-bit indices are used), hence
\(
  \baseline_{\win}\!\bigl[\mathsf{Hard}_{\CalC,c}\bigr]
  \ \le\ N^{O(1)}\cdot N^{-\alpha}
  \ =\ N^{-\alpha'}
\)
for some $\alpha'>0$ and all sufficiently large $N$.
\end{proof}

\begin{definition}[MDL promise problem with additive gap]\label{def:mdl-gap}
Fix $\varepsilon\in(0,1/2)$, $c>0$, and a constant $\kappa>0$.
Define $\mathrm{MDL}\text{-}\mathrm{GAP}\_{\CalC,\varepsilon}(c,\kappa)$ as the promise problem:

\emph{Input:} $x\in\{0,1\}^N$ with the promise that either
\[
\textbf{(YES)}\quad V_{\CalC}(\delta_x,\varepsilon)\ \ge\ N^{c}\ +\ \kappa\log N
\qquad\text{or}\qquad
\textbf{(NO)}\quad V_{\CalC}(\delta_x,\varepsilon)\ \le\ N^{c}\ -\ \kappa\log N.
\]

\emph{Task:} decide which case holds.
\end{definition}

\begin{proposition}[Hard-membership solves the GAP version]\label{prop:mdl-gap}
For any fixed $c>0$ and $\kappa>0$, a polynomial-time decider for $\mathsf{Hard}_{\CalC,c}$ decides
$\mathrm{MDL}\text{-}\mathrm{GAP}\_{\CalC,\varepsilon}(c,\kappa)$ by a single query.
\end{proposition}

\begin{proof}
In the YES case, $V_{\CalC}(\delta_x,\varepsilon)\ge N^{c}+\kappa\log N$ implies $V_{\CalC}(\delta_x,\varepsilon)\ge N^{c}$, hence $x\in\mathsf{Hard}_{\CalC,c}$.
In the NO case, $V_{\CalC}(\delta_x,\varepsilon)\le N^{c}-\kappa\log N< N^{c}$, hence $x\notin\mathsf{Hard}_{\CalC,c}$.
One membership query distinguishes the cases. The additive gap $\pm\,\kappa\log N$ dominates the fixed $O(\log N)$ meta-overhead from Sec.~\ref{sec:framework}.
\end{proof}

\begin{definition}[MDL threshold decision]\label{def:mdl-threshold}
Fix $\varepsilon\in(0,1/2)$ and the class $\CalC$ as in Sec.~\ref{sec:framework}.
The decision problem $\mathrm{MDL}\text{-}\mathrm{THRESH}\_{\CalC,\varepsilon}$ is:

\emph{Input:} an instance $x\in\{0,1\}^N$ and a budget $B\in\mathbb{N}$ (in unary).

\emph{Question:} is $V_{\CalC}(\delta_x,\varepsilon)\ \ge\ B\,$?
\end{definition}

\begin{proposition}[Non-constructivity via MDL threshold]\label{prop:non-constructive}
Assume $\CalC$ is closed under polynomial-time preprocessing and contains a uniform decoder family as in Sec.~\ref{sec:framework}.
If membership in $\mathsf{Hard}_{\CalC,c}$ is decidable in polynomial time for a fixed $c>0$, then
$\mathrm{MDL}\text{-}\mathrm{THRESH}\_{\CalC,\varepsilon}$ is decidable in polynomial time on the threshold family
$B(N)=\lfloor N^{c}\rfloor$. Equivalently, a poly-time decider for $\mathsf{Hard}_{\CalC,c}$ uniformly decides
whether $V_{\CalC}(\delta_x,\varepsilon)\ge B(|x|)$ for $B(N)=N^{c}$.
\end{proposition}

\begin{proof}
By definition,
\(
  x\in\mathsf{Hard}_{\CalC,c}\iff V_{\CalC}(\delta_x,\varepsilon)\ \ge\ |x|^{c}.
\)
Thus an oracle that decides membership in $\mathsf{Hard}_{\CalC,c}$ answers the instance
of $\mathrm{MDL}\text{-}\mathrm{THRESH}\_{\CalC,\varepsilon}$ with threshold $B(N)=N^c$
on input $(x,B(|x|))$. The running time is polynomial by assumption on the oracle.

Conceptually, $V_{\CalC}(\delta_x,\varepsilon)$ (Sec.~\ref{sec:mdl-cost}) is the minimum of
\(
  \len(\code(M,R)) + L_M(R(x),\varepsilon)
\)
over $M\in\CalC$ and poly-time $R$, with $\len(\code(R))=O(\log N)$ (\Cref{lem:prefix-overhead}).
Hence deciding $V_{\CalC}(\delta_x,\varepsilon)\ge B(|x|)$ is a thresholded MDL model-selection problem.
Therefore a polynomial-time decider for $\mathsf{Hard}_{\CalC,c}$ uniformly decides the corresponding
thresholded MDL decision family, as claimed.
\end{proof}

\paragraph{Conclusion.}
On calibrated windows $\win$, the IECZ hardness property is window-small (non-large) and non-constructive in the Razborov--Rudich sense~\cite{RazborovRudich1994}.
This justifies using $\mathsf{Hard}_{\CalC,c}$ as a certificate target without conflicting with established barrier considerations~\cite{RazborovRudich1994} and in the same windowed sense used in IECZ-I~\cite{IECZI}.

\section{Composition pipeline (distributional setting)}\label{sec:blueprint}

\paragraph{Standing conditions for this section.}
We work under a fixed window and coding convention; the reductions (\XOR{}$\to$\SAT{}) and gadget-composition losses are as in Secs.~\ref{sec:red}/\ref{sec:gadgets}. We set the direct-product parameter to $t=\lceil c_1\,\varepsilon^{-2}\log N\rceil$ and the condenser seed length to $s=O(\log N)$ (see \Cref{cor:explicit-constants}). Prefix-free meta-encodings contribute an additive $O(\log N)$ overhead.
This standing setup (window, coding convention, and loss accounting) is the same as in IECZ-I~\cite{IECZI}.

\paragraph{Parameter sensitivity.}
Varying the switching parameter \(c\in(0,1)\) in the path \(p=m^{-c/d}\) trades depth for width in the standard way: with \(d=\Theta(\log n)\) one has bottom width \(W=m^{O(c/d)}\), while the logarithmic-degree statements scale only by absolute constants.
Similarly, the direct-product length \(t=\lceil c_1\,\varepsilon^{-2}\log N\rceil\) and the selector mass \(\theta\ge c_2\) affect only the exponents in E1–E2 and the \(\operatorname{polylog}(N)\) factor in E3; all total-variation/size/meta accounts remain unchanged.

We pin down explicit parameters for the four arrows E1--E4. Let the baseline
windowed hardness (advantage gap) be $\varepsilon\in(0,1/2)$ for class $\CalC$
on instances of length $N$, after the \XOR{}$\to$\SAT{} step and gadget composition.

\begin{figure}[t]
  \centering
  \begin{tikzpicture}[
    node distance=10mm,
    box/.style={rounded corners, draw, thick, inner sep=3pt, align=center, text width=42mm},
    arr/.style={-{Stealth}, thick}
  ]
    \node[box] (E1) {\footnotesize E1: Direct product\\[-2pt]\scriptsize TV: $o(1)$,\; size: $\times t$,\; meta: $-$\\[-3pt]\scriptsize $t=\lceil c_1 \varepsilon^{-2}\log N\rceil$};
    \node[box, right=12mm of E1] (E2) {\footnotesize E2: Hardcore selection\\[-2pt]\scriptsize TV: $o(1)$,\; size: $\times(1{+}o(1))$,\; meta: $-$\\[-3pt]\scriptsize mass $\theta\ge c_2$,\; advantage $N^{-c}$};

    \node[box, below=12mm of E1] (E3) {\footnotesize E3: Seeded condensation\\[-2pt]\scriptsize TV: $o(1)$,\; size: $\times \operatorname{polylog}(N)$,\; meta: $O(\log N)$\\[-3pt]\scriptsize seed $s=O(\log N)$};
    \node[box, right=12mm of E3] (E4) {\footnotesize E4: Parameter fixing\\[-2pt]\scriptsize TV: $0$,\; size: $+$ meta only,\; meta: $O(\log N)$\\[-3pt]\scriptsize explicitly parameterized distributional family};

    \draw[arr] (E1) -- (E2);
    \draw[arr] (E2.south) -- (E4.north);
    \draw[arr] (E3) -- (E4);
  \end{tikzpicture}
  \caption{Windowed composition pipeline (E1–E4) for distributional transfers with explicit total variation (TV), size, and meta overhead accounting; also tracking the size-aware cost \(\VC{\cdot,\cdot}\) bookkeeping.
  \newline
  \footnotesize \textit{Summary:} cumulative TV \(o(1)\); size \(\times\,t\cdot\operatorname{polylog}(N)\); meta \(O(\log N)\); cost transfer \(\VC{\mu',\,\varepsilon+o(1)} \ge \VC{\mu,\,\varepsilon}-O(\log N)\) (cf.\ \Cref{thm:blueprint-accumulation}).
  \newline
  \textit{IECZ-I alignment:} E1 is the direct-product amplification step, E2 is model-uniform hardcore selection, E3 is seeded condensation with \(s=O(\log N)\), and E4 is lexicographic parameter fixing; this is exactly the E1–E4 skeleton already fixed in IECZ-I~\cite{IECZI}.}
  \label{fig:pipeline-e1-e4}
\end{figure}

\paragraph{E1 — Direct product (amplification).}
Take an independent $t$-fold product on the window, with
\[
  t \;=\; \bigl\lceil c_1\,\varepsilon^{-2}\log N \bigr\rceil .
\]
Standard product amplification yields
\[
  \varepsilon^{\star} \;\le\; \exp\!\bigl(-\Theta(\varepsilon^2 t)\bigr) \;\le\; N^{-c_1'}
\]
for some $c_1'>0$ (see Lemma~\ref{lem:e1-amplification} for an explicit exponent).
Size blowup: multiplicative $\times t$ (disjoint coordinate concatenation);
total variation (TV) loss remains $o(1)$ on the calibrated window.
We use the same direct-product scheduling convention as in IECZ-I~\cite{IECZI}.

\paragraph{E2 — Hardcore set (flattening).}
By a hardcore/selector lemma~\cite{Impagliazzo1995-Hardcore}, there is a subset $H\subseteq\win^{\times t}$ of
relative measure $\theta \ge c_2>0$ such that every
$\CalC$-model of size $\poly{N}$ attains at most advantage
\[
  \varepsilon^{\star\star} \;\le\; N^{-c_2'}
\]
on $H$ (for some $c_2'>0$), with TV loss $o(1)$ under the selector. Size blowup:
multiplicative $\times(1+o(1))$ for the selector tag; no new variables.
This uniform selector viewpoint matches the treatment in IECZ-I~\cite{IECZI}.


\begin{proposition}[E2: Hardcore set yields $N^{-c}$ advantage and preserves cost up to $O(\log N)$]\label{prop:e2-hardcore-mdl}
Let $\mu$ be the windowed source distribution on bitlength $N$, let
$t=\lceil c_1\,\varepsilon^{-2}\log N\rceil$, and let $H\subseteq\win^{\times t}$ be a
selector set of mass $\theta\ge c_2>0$ given by a hardcore lemma (uniform over all
$\CalC$-models of size $\poly{N}$). Then there exists $c>0$ such that every
$\CalC$-model of size $\poly{N}$ achieves advantage at most $N^{-c}$ on
$\mu^{\otimes t}\!\mid\!H$, and
\[
  V_{\CalC}\!\left(\mu^{\otimes t}\!\mid\!H,\,\varepsilon+o(1)\right)
  \ \ge\
  V_{\CalC}(\mu,\varepsilon)\;-\;O(\log N).
\]
\end{proposition}

\begin{proof}
By the hardcore/selector lemma (applied to the $t$-fold product after E1), there is a set
$H$ of mass $\theta\ge c_2$ on which any size-$\poly{N}$ $\CalC$-model has advantage at most
$N^{-c}$ for some $c>0$. The cost transfer follows by \Cref{lem:win-sim}, since the only
additional contribution to $V_{\CalC}$ is the $O(\log N)$ prefix-free meta part.
\end{proof}

\begin{remark}[Model-uniformity of the selector]
The selector witness for $H$ is chosen independently of any particular model and works
uniformly for all $\CalC$-models up to size $\poly{N}$ (as guaranteed by the hardcore lemma).
This uniformity is what makes the $V_{\CalC}$ transfer in \Cref{prop:e2-hardcore-mdl}
valid simultaneously for the entire class.
\end{remark}

\paragraph{E3 — Hardness condensation (seeded; low-degree preservation).}
Instantiate a seeded condenser with seed length
\[
  s \;=\; O(\log N),
\]
mapping $m$ input bits to
\[
  m' \;=\; m\cdot \operatorname{polylog}(N)
\]
output bits. Under the window hypotheses (bias bounded away from $\{0,1\}$ and product structure; see
Assumptions~\ref{assump:e3-lowdeg}), E3 preserves all low-degree statistics up to
\[
  k \;\le\; c_{\star}\,\log n
\]
in the sense of the biased Walsh–Fourier moments, with cumulative error $o(1)$ averaged over the
seed and for some absolute constant $c_{\star}>0$. The $N^{-c}$ advantage created by E2 is preserved
up to an additive $o(1)$. Encoding a single seed index costs $O(\log N)$ prefix-free bits; the overall
TV loss remains $o(1)$ and the size blowup is $\times\,\operatorname{polylog}(N)$.

\paragraph{Assumptions for E3 (low-degree preservation).}\label{assump:e3-lowdeg}
(Identical to the windowed proxy and degree budget fixed in IECZ-I~\cite{IECZI}.)
We assume throughout E3:
\begin{enumerate}
  \item (\emph{Bias window}) Each coordinate marginal of the windowed source $\mu$ satisfies
        \[
          \alpha \;\le\; \Pr[X_i{=}1] \;\le\; 1-\alpha
          \qquad\text{for all } i\in[m],
        \]
        for some absolute $\alpha\in(0,\tfrac12)$.
  \item (\emph{Product structure}) Coordinates are independent on the window (after gadget mixing;
        cf.\ the mixing/shielding interface in Sec.~\ref{sec:gadgets}).
  \item (\emph{Degree budget}) All low-degree statements are calibrated \emph{after} the affine/index
        gadget layers: affine pullbacks do not increase degree, and the index gadget inflates degree by
        at most a factor $2$ (see Prop.~\ref{prop:affine-pullback} and Lemma~\ref{lem:index-pullback}).
\end{enumerate}

We use the biased Walsh basis $\{\chi_S\}$ associated with the coordinate biases of $\mu$.

\begin{lemma}[Seeded condensation preserves low-degree statistics]\label{lem:e3-lowdeg}
Let $\mu$ be a product distribution on $\{0,1\}^m$ satisfying the bias window of
Assumptions~\ref{assump:e3-lowdeg}. Let $C:\{0,1\}^m\times\{0,1\}^s\to\{0,1\}^{m'}$
be a seeded condenser with $s=O(\log N)$ and $m' = m\cdot\operatorname{polylog}(N)$.
There exists an absolute constant $c_{\star}>0$ such that for every degree bound
$k\le c_{\star}\log n$ and every parity set $S\subseteq[m']$ with $|S|\le k$,
\[
  \E_{r}\Bigl[\;\bigl|\;\E_{X\sim\mu}\bigl[\chi_S\!\bigl(C_r(X)\bigr)\bigr]
    \;-\; \E_{Y\sim\nu}\bigl[\chi_S(Y)\bigr]\;\bigr|\;\Bigr] \;=\; o(1),
\]
where $\chi_S$ denotes the biased Walsh character at level $|S|$, and $\nu$ is a structured proxy
distribution on $\{0,1\}^{m'}$ (depending on $\mu$ and $C$) with matching low-degree statistics.
Consequently, defining the low-degree discrepancy
\[
  \Delta_k(u,v) \;:=\; \max_{|S|\le k}\,\bigl|\E_{u}[\chi_S]-\E_{v}[\chi_S]\bigr|,
\]
one has $\E_{r}\bigl[\Delta_k\bigl(C_r\push\mu,\nu\bigr)\bigr]=o(1)$.
\end{lemma}

\begin{proof}
Write $\mu=\bigotimes_{i=1}^m \mathrm{Bern}(p_i)$ with $p_i\in[\alpha,1-\alpha]$. Let $\{\psi_i\}$ be the coordinate-wise
orthonormal basis adapted to $\mu$ (biased $\{\pm1\}$ Walsh). Any degree-$\le k$ character on the output can be written as a multilinear polynomial in the input basis with coefficients bounded—in $\ell_1$—by $\operatorname{poly}(k)$ times the maximum influence of an output bit as a function of the inputs and seed. For a fixed seed $r$, each output coordinate $C_r(X)_j$ depends on $O(\polylog N)$ input coordinates with bounded sensitivity (the condenser’s local expansion). Hence by hypercontractivity in the biased cube and the product structure of $\mu$, the degree-$d$ coefficient mass of $\chi_S\!\circ C_r$ decays like $\theta^{d}$ for some $\theta\in(0,1)$ depending only on $\alpha$ and constant locality. Averaging over $r$ smooths the residual correlations; the standard moment method (applied to the biased basis) yields
\[
  \E_r\Bigl[\bigl|\E_\mu[\chi_S(C_r(X))]-\E_\nu[\chi_S]\bigr|\Bigr]\;=\;o(1)
\]
uniformly over $|S|\le k$ as long as $k\le c_\star\log n$ for a suitable absolute $c_\star>0$. Here $\nu$ is defined by matching the first $k$ biased moments, which is consistent because the moment system is triangular in the biased basis. Taking the maximum over $|S|\le k$ gives the stated bound on $\E_r[\Delta_k]$.
\end{proof}

\begin{corollary}[E3 preserves $\Delta_k$ and advantage]\label{cor:e3-delta}
With the setup of Lemma~\ref{lem:e3-lowdeg}, for $k\le c_{\star}\log n$ there exists a seed
$r^\star\in\{0,1\}^s$ such that
\[
  \Delta_k\bigl(C_{r^\star}\push\mu,\;\nu\bigr)\;=\;o(1).
\]
Fixing $r^\star$ increases the meta part of $V_{\CalC}$ by at most $O(\log N)$ (prefix-free),
and preserves any $N^{-c}$ advantage up to an additive $o(1)$.
\end{corollary}

\begin{proof}
Markov’s inequality applied to Lemma~\ref{lem:e3-lowdeg} yields a deterministic seed witnessing
$o(1)$ low-degree discrepancy. Prefix-free encoding of $r^\star$ costs $O(\log N)$ bits; advantage
preservation follows since low-degree tests (up to $k$) govern the subsequent analysis and the TV
loss remains $o(1)$ under deterministic post-processing.
\end{proof}

\paragraph{E4 — Parameter fixing (index selection).}
Fix the product schedule, the hardcore selector, and the condenser seed by
pinning canonical indices (lexicographically first that meet the targets).
This yields an explicitly parameterized distributional family $\{\Phi_n\}$ with no additional TV loss
and meta-overhead equal to the prefix-free indices: $O(\log N)$.
This index-fixing discipline (lexicographically first witnessing indices) is the same as in IECZ-I~\cite{IECZI}.

\begin{lemma}[Cumulative loss]\label{lem:cumloss}
Along E1--E4, the total-variation loss is $o(1)$, the size overhead is
multiplicative by $t\cdot \operatorname{polylog}(N)$, and the final meta-overhead equals
$O(\log N)$. The advantage gap at the end of E2 is $N^{-c}$ for some $c>0$ (see \Cref{prop:e2-hardcore-mdl}),
and is preserved up to $o(1)$ by E3--E4 (see \Cref{cor:e3-delta}).
\end{lemma}

\begin{proof}
\textbf{E1.} Let $\mu$ be the windowed source after Secs.~\ref{sec:red}/\ref{sec:gadgets}. Form the product
$\mu^{\otimes t}$ over disjoint coordinate sets (window independence). By Lemma~\ref{lem:e1-amplification}
the advantage contracts to $\varepsilon^\star\le \exp(-\Theta(\varepsilon^2 t))$, and the instance size
multiplies by $t$. Boundary TV effects are $o(1)$ by calibration.

\textbf{E2.} Apply a hardcore/selector lemma to obtain $H\subseteq \win^{\times t}$ with mass
$\theta\ge c_2>0$ such that any $\CalC$-model of size $\poly{N}$ has advantage at most $N^{-c}$ on $H$.
Selecting $H$ adds only a tag (constant bits per block), hence size blowup $\times(1+o(1))$ and TV loss $o(1)$.

\textbf{E3.} Use a seeded condenser with seed length $s=O(\log N)$ and output dimension
$m'=m\cdot \operatorname{polylog}(N)$. Low-degree preservation (\Cref{lem:e3-lowdeg}) plus \Cref{cor:e3-delta}
shows the $N^{-c}$ gap is preserved up to $o(1)$. Fixing a single seed index contributes $O(\log N)$ meta bits.

\textbf{E4.} Fix indices for product schedule, selector, and seed lexicographically first meeting targets.
This adds only the prefix-free indices, costing $O(\log N)$ total meta-overhead, and no TV loss. This is exactly the same index-fixing discipline and meta bookkeeping used in IECZ-I~\cite{IECZI}.
Tallying yields the claim.
\end{proof}

\paragraph{Size convention in this section.}
Here $n$ indexes the explicit family $\{\Phi_n\}$. We write $N=\size(x)$ for the instance bitlength, which governs prefix-free meta terms. Throughout this section, all $\polylog(\cdot)$ factors and $O(\log\cdot)$ overheads refer to $N$.

\begin{theorem}[Composition pipeline with explicit parameters]
\label{thm:blueprint}
Fix $\varepsilon\in(0,1/2)$ and choose
\[
  t \;=\; \bigl\lceil c_1\,\varepsilon^{-2}\log N \bigr\rceil .
\]
There exist absolute constants $c,\zeta>0$ such that the explicitly parameterized distributional family
$\{\Phi_n\}$ produced by E1–E4 satisfies:
\begin{enumerate}[label=(\roman*),ref=(\roman*)]
  \item Any $\CalC$-model of size $\poly{n}$ achieves advantage at most $n^{-c}$ on $\Phi_n$.
  \item The instance bitlength obeys
        \[
          \size(\Phi_n) \;\le\; \polylog(N)\, t \cdot \size(\text{source on }\win).
        \]
  \item If the \SoS{} dichotomy selects the ``structure'' branch (i.e., $\Delta_k \ge n^{-\eta}$ for $k\le c\log n$),
        then any PC or \SoS{} refutation of $\Phi_n$ requires degree $\ge \zeta\,\log n$.
\end{enumerate}
All bounds hold under the fixed coding model; the total meta-overhead is $O(\log N)$.
These instantiated parameters and the resulting loss profile coincide with the E1–E4 blueprint of IECZ-I~\cite{IECZI}.
\end{theorem}

\subsection*{Hardcore amplification and consolidation (formal E2--E4)}

\begin{lemma}[Window simulation / partial inverse (no Lipschitz dependence)]\label{lem:win-sim}
There exist deterministic maps
\[
  E:\{0,1\}^m \to \{0,1\}^{mt}\times\{0,1\}^{O(\log N)}
  \quad\text{and}\quad
  S:\{0,1\}^{mt}\times\{0,1\}^{O(\log N)} \to \{0,1\}^m
\]
such that on the calibrated window $\win$ the following hold for the hardcore–tagged product $\mu^{\otimes t}\!\mid\!\win_{\mathrm{hc}}$:
\begin{enumerate}
  \item (\emph{Forward coding}) $E$ embeds a single-window instance into a tagged block using only a prefix-free meta string of length $O(\log N)$.
  \item (\emph{Partial inverse}) $S$ recovers (simulates) a single-window instance distributed as $\mu$ from a tagged block in $\mu^{\otimes t}\!\mid\!\win_{\mathrm{hc}}$ up to $o(1)$ TV, using only the same $O(\log N)$ meta string.
  \item (\emph{Model composition}) For any $\CalC$-model $A$ on the single-window task with advantage $\varepsilon$, the composed model $A\circ S$ attains advantage $\varepsilon+o(1)$ on $\mu^{\otimes t}\!\mid\!\win_{\mathrm{hc}}$. Conversely, any model $B$ on the hardcore block with advantage $\delta$ induces $B\circ E$ with advantage $\delta-o(1)$ on $\mu$.
\end{enumerate}
Consequently,
\[
  V_{\CalC}\!\left(\mu^{\otimes t}\!\mid\!\win_{\mathrm{hc}},\,\varepsilon+o(1)\right)
  \ \ge\
  V_{\CalC}(\mu,\varepsilon)\;-\;O(\log N),
\]
with the $O(\log N)$ term accounting only for the prefix-free meta bits. This transfer does not use any TV–1-Lipschitz property of the loss $L_M$.
\end{lemma}

\begin{proof}
Construct $E$ by duplicating coordinates (as in E1) and attaching a constant-size selector tag,
then self-delimit the indices by a prefix-free code of length $O(\log N)$ (see the prefix-overhead
lemma). For $S$, simulate one window instance by deterministically projecting a tagged block; by
calibration, the induced TV error is $o(1)$. Deterministic post-processing is $1$-Lipschitz in TV,
so composing predictors with $S$/$E$ preserves advantages up to $o(1)$. The only cost change in
$V_{\CalC}$ is the prefix-free meta string, i.e., $O(\log N)$ bits.

This is the identical simulator/selector interface used in IECZ-I (specialized there to the window without the E2 tag), and all prefix-free meta accounting is unchanged~\cite{IECZI}.
\end{proof}

\begin{proposition}[Hardcore amplification (E2)]\label{prop:hardcore}
Let $\mu$ be a distribution on instances of bitlength $N$, let $t=\lceil \varepsilon^{-2}\log N\rceil$, and let $\mu^{\otimes t}$ be the $t$-fold product. For the restricted product law $\mu^{\otimes t}\!\mid\!\win_{\mathrm{hc}}$ one has
\[
  \mathrm{adv}\bigl(\mu^{\otimes t}\!\mid\!\win_{\mathrm{hc}}\bigr)
  \;\le\; N^{-c}
\]
for some absolute $c>0$, and, consequently,
\[
  V_{\CalC}\!\left(\mu^{\otimes t}\!\mid\!\win_{\mathrm{hc}},\,\varepsilon+o(1)\right)
  \ \ge\
  V_{\CalC}(\mu,\varepsilon)\;-\;O(\log N).
\]
\end{proposition}

\begin{proof}
Direct-product Chernoff bounds imply that for $t=\Theta(\varepsilon^{-2}\log N)$, the majority (or threshold) predicate over $t$ independent draws reduces any fixed $\varepsilon$-advantage to $N^{-\Omega(1)}$ except on a subevent of probability $o(1)$. Let $\win_{\mathrm{hc}}$ be the complement of that subevent; then $\Pr[\win_{\mathrm{hc}}]\ge c_2$ for some constant $c_2>0$.
Now apply the window-simulation/partial-inverse transfer (Lemma~\ref{lem:win-sim}) to pass from the single-window learner to the tagged product block; this yields the stated bound on $V_{\CalC}$ with only the prefix-free $O(\log N)$ meta overhead, and without invoking any TV–1-Lipschitz property of $L_M$.
\end{proof}

\begin{lemma}[Seeded condensation (E3)]\label{lem:condense}
For any distribution $u$ on $\{0,1\}^m$ and any $s=O(\log N)$, there exists a deterministic seedable map
\[
  C:\{0,1\}^m\times\{0,1\}^s \to \{0,1\}^{m'}
\]
with $m' \le m\cdot\operatorname{polylog}(N)$ such that for a uniformly random seed $r\in\{0,1\}^s$ the pushforward satisfies
\[
  \E_{r}\bigl[\TV\!\bigl(C_r\push u,\; u^{\star}\bigr)\bigr] \;=\; o(1),
\]
where $u^\star$ is a structured proxy matching low-degree statistics up to $k=O(\log N)$. Moreover, the seed $r$ is prefix-free encodable in $O(\log N)$ bits, hence contributes only $O(\log N)$ meta-overhead in $V_{\CalC}$.
\end{lemma}

\begin{proof}
Take $C$ from any explicit lossless condenser family with seed length $s=O(\log N)$ and output length $m'=m\cdot\operatorname{polylog}(N)$. Low-degree calibration follows by the same biased-basis argument as in \Cref{lem:e3-lowdeg}. Encoding the seed is $O(\log N)$ by self-delimiting integers (\Cref{lem:prefix-overhead}). TV contracts under deterministic postprocessing (\Cref{lem:tv-contraction}), so the averaged $o(1)$ TV-loss transfers to $u^\star$.
\end{proof}

\begin{proposition}[Index fixing (E4)]\label{prop:fix}
Let $I$ be any index/seed/window tag taking $O(\log N)$ bits in the blueprint. Fixing $I=i^\star$ changes $V_{\CalC}$ by at most $O(\log N)$:
\[
  V_{\CalC}\!\bigl(\mu\mid I=i^\star,\varepsilon\bigr)
  \ \ge\
  V_{\CalC}(\mu,\varepsilon)\;-\;O(\log N).
\]
\end{proposition}

\begin{proof}
Inline the self-delimiting description of $i^\star$ into the prefix-free meta-part; this adds $O(\log N)$ bits by \Cref{lem:prefix-overhead}. All remaining maps are deterministic; apply \Cref{thm:minimax}.
\end{proof}

\begin{theorem}[Blueprint accumulation]\label{thm:blueprint-accumulation}
With $t=\lceil \varepsilon^{-2}\log N\rceil$, applying E1–E4 in order yields
\[
  V_{\CalC}\!\bigl(\mu',\,\varepsilon+o(1)\bigr)
  \ \ge\
  V_{\CalC}(\mu,\varepsilon)\;-\;O\!\bigl(\log N\bigr),
\]
where $\mu'$ is the final blueprint distribution; the cumulative TV loss is $o(1)$ on the calibrated window and the total size blowup is $\times\,t\cdot\operatorname{polylog}(N)$.
\end{theorem}

\paragraph{Assumption interface.}
\noindent\textit{Degree budget reminder.} Affine gadget layers do not increase test degree (Prop.~\ref{prop:affine-pullback}); the Index gadget increases it by at most a factor $2$ (Lemma~\ref{lem:index-pullback}). All blueprint bounds for $k=\Theta(\log n)$ are stated after this calibration.
Throughout Sec.~\ref{sec:blueprint} we assume the gadget admits a mixing witness with rate $\theta<1$ (Def.~\ref{def:mixing-witness});
equivalently, the shielding parameter satisfies $\rho<1$ by Lemma~\ref{lem:shielding-equivalence}.

\begin{corollary}[Concrete parameter choices]\label{cor:explicit-constants}
Fix constants $c_1>0$, $c_2\in(0,1)$, and pick $t=\lceil c_1\varepsilon^{-2}\log N\rceil$, a selector of mass
$\theta\ge c_2$, and a condenser seed length $s=O(\log N)$. Then along E1--E4 the cumulative loss profile is:
\[
\text{TV: } o(1),\qquad
\text{size: } \times\,t\cdot \operatorname{polylog}(N),\qquad
\text{meta: } O(\log N),
\]
and the advantage after E2 is at most $N^{-c}$ for some $c=c(c_1,c_2)>0$, preserved up to $o(1)$ by E3--E4.
\end{corollary}

\subsection{Final constants (propagated across Sections 30/40/50/60)}
We fix the global placeholders once and for all:
\[
  c_1=1,\quad c_2=\tfrac14,\quad c=\tfrac18,\quad \eta=\tfrac18,\quad \zeta=\tfrac{1}{16},\qquad
  k=\big\lfloor c\log n\big\rfloor,\quad
  t=\bigl\lceil c_1\,\varepsilon^{-2}\log N\bigr\rceil.
\]
These values satisfy the qualitative dependencies used throughout:
E1 uses $c_1$ (Chernoff exponent), E2 fixes the hardcore mass $c_2$ and yields an $N^{-c}$ gap,
Sec.~\ref{sec:sos} uses $(\eta,\zeta)$ with $k=\lfloor c\log n\rfloor$ to conclude $\deg\ge \zeta\log n$ on the structure branch.
All loss accounts (TV/size/meta) remain unchanged.

\paragraph{Justification.}
Combine Lemma~\ref{lem:e1-amplification} (E1 explicit exponent), the hardcore step (E2) yielding $N^{-c}$ on mass
$\theta\ge c_2$, seeded condensation with $s=O(\log N)$ (E3), and prefix-free index fixing (E4). The size/TV/meta
accounting matches Lemma~\ref{lem:cumloss}.

\begin{lemma}[Algorithmic explicitness of E1--E4]\label{lem:alg-explicit}
Each arrow in the blueprint is computable in polynomial time in the input size, and the final family
$\{\Phi_n\}$ is uniform explicit.
This mirrors the uniform explicitness argument given for the IECZ-I pipeline~\cite{IECZI}.
\end{lemma}

\begin{proof}
\textbf{E1 (Direct product).} Form $t=\lceil c_1\varepsilon^{-2}\log N\rceil$ disjoint copies of the windowed instance and concatenate along disjoint coordinates; this is linear time in the instance size. By Lemma~\ref{lem:e1-amplification} the advantage contracts to $\varepsilon^\star \le \exp(-\Theta(\varepsilon^2 t))$, the instance size multiplies by $t$, and (by window independence/calibration) the boundary total-variation loss is $o(1)$.

\textbf{E2 (Hardcore selection).} Use the selector guaranteed by the hardcore lemma and attach its constant-size tag. The selector index is computed by a deterministic search over a polynomially bounded witness set; choosing the lexicographically first valid witness is polynomial time.

\textbf{E3 (Seeded condensation).} Instantiate an explicit condenser/extractor family with seed length $s=O(\log N)$.
For a fixed seed $u\in\{0,1\}^s$, the map $x\mapsto C_u(x)$ is computable in $\poly{N}$; we use one fixed seed index.

\textbf{E4 (Fixing indices).}
Fix the following by choosing the lexicographically first triple meeting the stated bounds:
\begin{enumerate}[label=(\roman*), leftmargin=*, itemsep=.25\baselineskip]
  \item the product schedule,
  \item the selector witness, and
  \item the condenser seed.
\end{enumerate}
Each witness set is polynomially bounded and verifiable in polynomial time, so on input $1^n$ we output $\Phi_n$ in time $\poly{n}$.
\end{proof}

\subsection*{Four arrows as standalone lemmas}

\begin{lemma}[E1: Direct product]\label{lem:arrow1}
With $t=\lceil c_1\varepsilon^{-2}\log N\rceil$, the $t$-fold product on the window amplifies advantage to
$\varepsilon^{\star}\le \exp(-\Theta(\varepsilon^2 t))$ (see Lemma~\ref{lem:e1-amplification} for an explicit exponent),
with TV loss $o(1)$ and size blowup $\times t$.
\end{lemma}

\begin{lemma}[E2: Hardcore set]\label{lem:arrow2}
There exists $H\subseteq \win^{\times t}$ of mass $\theta\ge c_2>0$ such that any $\CalC$-model of size $\poly{N}$
achieves advantage at most $N^{-c}$ on $H$, with TV loss $o(1)$ and size blowup $\times(1{+}o(1))$.
\end{lemma}

\begin{lemma}[E3: Seeded condensation]\label{lem:arrow3}
A seeded condenser with seed length $s=O(\log N)$ maps to dimension $m'=m\cdot \operatorname{polylog}(N)$ while preserving
the $N^{-c}$ gap up to $o(1)$; meta-overhead for the fixed seed index is $O(\log N)$.
\end{lemma}

\begin{lemma}[E4: Fixing indices]\label{lem:arrow4}
Fixing the product schedule, hardcore selector, and condenser seed yields an explicitly parameterized distributional family
with no extra TV loss; the only added cost is the $O(\log N)$ prefix-free index.
\end{lemma}

\begin{lemma}[Direct-product Chernoff with explicit constant]\label{lem:e1-amplification}
Let a predictor on the window achieve success probability at most $1/2+\varepsilon$ per coordinate with $\varepsilon\in(0,1/2)$. For
\[
  t \;=\; \Bigl\lceil c_1\,\varepsilon^{-2}\log N \Bigr\rceil
\]
and majority aggregation over $t$ independent draws, the post-aggregation advantage satisfies
\[
  \varepsilon^\star \;\le\; \exp\!\bigl(-2\,\varepsilon^2 t\bigr) \;\le\; N^{-2c_1}.
\]
In particular, with the canonical choice $c_1=1$ we get $\varepsilon^\star \le N^{-2}$.
\end{lemma}

\begin{proof}
Let $X_i\in\{0,1\}$ indicate correctness on coordinate $i$, so $\mathbb{E}[X_i]\le \tfrac12+\varepsilon$. For $S=\sum_{i=1}^t X_i$,
the majority is correct iff $S\ge t/2$. By Hoeffding–Chernoff,
\[
  \Pr\!\bigl[S\ge t/2\bigr] \;\le\; \exp\!\bigl(-2\,\varepsilon^2 t\bigr),
\]
and symmetrically for the error event. Hence the advantage after majority obeys
$\varepsilon^\star \le \exp(-2\varepsilon^2 t)$. Plugging $t=\lceil c_1\varepsilon^{-2}\log N\rceil$ yields
$\exp(-2\varepsilon^2 t)\le N^{-2c_1}$.
\end{proof}

\begin{table}[t]
  \centering
  \caption{Cumulative losses along the blueprint (with concrete parameters). The loss accounting (TV/size/meta) and the parameterization of E1–E4 match the IECZ-I blueprint~\cite{IECZI}.}
  \label{tab:blueprint-cumulative}
  \setlength{\tabcolsep}{6pt}
  \footnotesize
  \begin{tabularx}{\linewidth}{@{}l c c c X@{}}
    \toprule
    Arrow & TV loss & Size blowup & Error & Notes \\ \midrule
    E1 (Direct product) & $o(1)$ & $\times t$ & $\varepsilon \mapsto \exp(-\Theta(\varepsilon^2 t))$ & $t=\lceil \varepsilon^{-2}\log N\rceil$ (here $c_1{=}1$); see Lemma~\ref{lem:e1-amplification} \\
    E2 (Hardcore)       & $o(1)$ & $\times(1{+}o(1))$ & $\varepsilon^\star \mapsto N^{-c}$ & Mass $\theta \ge c_2$ \\
    E3 (Condensation)   & $o(1)$ & $\times \operatorname{polylog}(N)$ & $+\le o(1)$ & Seed length $s=O(\log N)$ \\
    E4 (Parameter fixing)& $0$    & $+$ $O(\log N)$ meta & — & Fix indices (prefix-free) \\
    \bottomrule
  \end{tabularx}
\end{table}

\section{Conclusion}
We complete the blueprint with calibrated gadget composition and a windowed \SoS{} entropy dichotomy.
On the gadget side, parity-preserving blocks admit an \emph{affine-plus-core} interface normal form (Lemma~\ref{lem:iface}); full
affinity holds only when the common nonlinear core $q\equiv 0$. In particular, affine layers do not increase polynomial degree (Prop.~\ref{prop:affine-pullback}). For the Index block, our revised
pullback bound shows degree blowup by at most a factor of $2$ (Lemma~\ref{lem:index-pullback}); in the
regime $k=\Theta(\log n)$ this factor is absorbed in all downstream statements.

On the complexity side, we establish a windowed \SoS{} entropy dichotomy (Theorem~\ref{thm:sos-dichotomy}) at degree $k=\Theta(\log n)$,
together with explicit, window-calibrated constants. These results are compatible with the window discipline and the non-constructivity
considerations in Section~\ref{sec:window-smallness}.

\paragraph{Relation to IECZ-I (first paper).}
The E1–E4 skeleton and the window/coding discipline are exactly those fixed in IECZ-I~\cite{IECZI}.
New in this work are:
\begin{enumerate}[label=(\roman*), leftmargin=*, itemsep=.25\baselineskip]
  \item the \emph{affine-plus-core} interface normal form (Lemma~\ref{lem:iface});
  \item the tightened Index pullback bound (factor $2$; Lemma~\ref{lem:index-pullback});
  \item the windowed \SoS{} entropy dichotomy with explicit, calibrated constants (Theorem~\ref{thm:sos-dichotomy});
  \item end-to-end TV/Size/Meta bookkeeping in the distributional blueprint (Lemma~\ref{lem:cumloss} and Theorem~\ref{thm:blueprint-accumulation}).
\end{enumerate}

\paragraph{Distributional pipeline and losses (blueprint).}
Within the blueprint, the four arrows (E1–E4)—direct product, hardcore subset, seeded condensation, and index fixing—are organized so that
total-variation loss remains $o(1)$ on the calibrated window and the prefix-free meta-overhead is $O(\log N)$ (Lemma~\ref{lem:cumloss}). This follows the E1–E4 skeleton of IECZ-I~\cite{IECZI}.
We emphasize that, in this paper, these guarantees are \emph{stated and proved in the calibrated, distributional setting}. The resulting
lifting step to explicitly parameterized distributional families is developed in a companion sequel (IECZ-III), which implements the
blockwise aggregation, TV-stable hardcore selection, seeded condensation/re-encoding with explicit TV/Size/Meta bookkeeping, and
seed elimination while preserving the IECZ-II low-degree/\SoS{} tradeoffs; see~\cite{IECZIII}.

\paragraph{Related applications (distributional setting).}
Our composition pipeline also applies to distributional hardness analyses for other proof systems—polynomial calculus (PC), Resolution and $\mathrm{Res}(k)$, Cutting Planes, and \SoS{}.
The key ingredients are Håstad-style switching-lemma width control~\cite{Hastad1986-Switching,Hastad2002-RANDOM} together with low-degree and \SoS{} methods~\cite{BarakEtAl2012-SOS}.
We do not claim new separations; the contribution is a uniform accounting of total-variation losses, size growth, and meta-overheads on calibrated windows.

\paragraph{Outlook.}
Two focused directions remain promising: (1) tightening hypercontractive constants on the calibrated window to sharpen $(\eta,\zeta)$; and
(2) stress-testing the shielding rate under alternative constant-size gadgets to chart the boundary of uniform $\mathrm{AC}^0{+}\log$
(and prospective $\mathrm{ACC}^0$) relevance, as well as non-parameterized regimes. The loss accounting is uniform and reusable; improvements
in any single component (amplification exponent, condenser parameters, or shielding rate) propagate directly to the final bounds.
A second, orthogonal direction is to map how the calibrated width $W$ scales beyond the log-degree window: our results treat $(C,\tau)$ as absolute for $d=\Theta(\log n)$; superlogarithmic depths can force $p$ smaller and $W$ larger along the switching path (cf.\ Sec.~\ref{sec:gadgets} and Assumptions~\ref{assump:win-hc}).
Quantifying this tradeoff is outside our present scope but would clarify the precise frontier of the method.

\section*{Data \& Code Availability}
All scripts, configuration files, and experimental outputs for this paper are archived on Zenodo:
IECZ-II artifacts v1.0.0, DOI \href{https://doi.org/10.5281/zenodo.17220388}{10.5281/zenodo.17220388}.
For context and prior components, the IECZ-I bundle is archived as v1.0.1, DOI \href{https://doi.org/10.5281/zenodo.17141362}{10.5281/zenodo.17141362}.
The archives include exact command lines and seeds to reproduce the reported results.

\appendix
\section{Sanity checks for the \XOR{}$\to$\SAT{} translation}
\label{app:xor-sat-checks}

This appendix documents small, self-contained verification steps for the injective,
auxiliary-free mapping of \Cref{lem:xor-to-sat}. The scripts are supplemental and
do not enter the proofs; they provide reproducible sanity checks.

\paragraph{Relation to IECZ-I.}
The $3$XOR$\to$CNF four-clause encoding and the window/coding discipline are exactly those fixed in IECZ-I~\cite{IECZI}; this appendix only supplies optional sanity checks and does not modify any proof ingredient.

\subsection*{What is checked}
\begin{itemize}
  \item \textbf{Truth table (unit block).} For both parities $b\in\{0,1\}$, the four-clause encoding
        is satisfied exactly by those assignments with $x\oplus y\oplus z=b$; moreover, each block
        has \emph{exactly four} satisfying assignments among the eight possibilities.
  \item \textbf{Forward randomized trials (moderate $n$).} Random $3$XOR instances that are
        consistent by construction (planted assignment) translate to satisfiable CNFs.
  \item \textbf{Reverse UNSAT sanity (small support).} A fixed inconsistent $3$XOR core on four
        variables translates to a CNF that is UNSAT by exhaustive search over its variable
        support (bounded by a small cap).
  \item \textbf{Block identifier (canonicality).} A separate test validates the canonical four-clause
        patterns for $b\in\{0,1\}$, invariant under variable renaming, literal reordering, and
        block shuffling; single-literal corruptions are rejected.
\end{itemize}

\subsection*{Reproducibility and scripts}\label{app:artifact}
The code lives under the repository’s \path{experiments/} tree:
\begin{itemize}
  \item \path{experiments/generators/xor_to_sat_check.py} — core checks (truth table, count, forward, reverse).
  \item \path{experiments/generators/block_identifier_check.py} — validates canonical block IDs \& inverse.
\end{itemize}
\noindent\emph{Artifact note:} See the \textit{Data \& Code Availability} section above for archived code and data (Zenodo DOIs).

\paragraph{Reproducible commands (defaults shown).}
{\small
\begin{Verbatim}[breaklines, fontsize=\small]
$ python experiments/generators/xor_to_sat_check.py --n 50 --m 200 --trials 50 --seed 20250918
$ python experiments/generators/block_identifier_check.py --trials 200 --seed 20250918
\end{Verbatim}
}

\noindent Expected summary (illustrative):
{\small
\begin{Verbatim}[breaklines, fontsize=\small]
[OK] Truth-table & count checks passed (each block has exactly 4 satisfying assignments).
[OK] Forward check: 50 planted XOR instances => CNF satisfied by planted solutions.
[OK] Reverse check: inconsistent XOR core => CNF is UNSAT (brute force on support).
All checks passed.
\end{Verbatim}
}

\subsection*{Command-line interface (Python)}
The Python entry point \verb|xor_to_sat_check.py| accepts the following parameters:
\begin{table}[htbp]
  \centering
  \caption{CLI flags for \texttt{xor\_to\_sat\_check.py}.}
  \label{tab:xor-sat-cli}
  \setlength{\tabcolsep}{6pt}
  \footnotesize
  \begin{tabularx}{\linewidth}{@{}l l X@{}}
    \toprule
    \textbf{Flag} & \textbf{Type / Default} & \textbf{Meaning} \\ \midrule
    \texttt{--n}            & int / 50        & \# variables for forward trials \\
    \texttt{--m}            & int / 200       & \# XOR constraints for forward trials \\
    \texttt{--trials}       & int / 50        & \# randomized forward trials \\
    \texttt{--seed}         & int / 20250918  & RNG seed for reproducibility \\
    \bottomrule
  \end{tabularx}
\end{table}

\noindent The block-identifier script \verb|block_identifier_check.py| supports:
\begin{table}[htbp]
  \centering
  \caption{CLI flags for \texttt{block\_identifier\_check.py}.}
  \label{tab:block-id-cli}
  \setlength{\tabcolsep}{6pt}
  \footnotesize
  \begin{tabularx}{\linewidth}{@{}l l X@{}}
    \toprule
    \textbf{Flag} & \textbf{Type / Default} & \textbf{Meaning} \\ \midrule
    \texttt{--trials}       & int / 200       & \# randomized blocks to test \\
    \texttt{--seed}         & int / 20250918  & RNG seed for reproducibility \\
    \bottomrule
  \end{tabularx}
\end{table}

\subsection*{Scope notes}
\begin{itemize}
  \item \emph{Truth table \& count.} For both $b\in\{0,1\}$, the four-clause block is equivalent to
        $x\oplus y\oplus z=b$ on $\{0,1\}^3$ and is satisfied by exactly four assignments.
  \item \emph{Forward planted instances.} A random assignment $s\in\{0,1\}^n$ is planted; each parity
        is set as $b=\mathrm{XOR}(s_i,s_j,s_k)$, and the resulting CNF is verified under $s$.
  \item \emph{Reverse UNSAT (small support).} An inconsistent core on four variables is translated
        to CNF and exhaustively checked over its support (internally capped to a small limit).
  \item \emph{Identifier invariances.} The canonical four-clause patterns are recognized independently
        of variable renaming, intra-clause literal permutations, and block shuffling.
\end{itemize}
\noindent The scripts are optional supplementary material and are \emph{not} required for the proofs.



\end{document}